%% file: main.tex
\documentclass[notitlepage,11pt]{article}
\pdfoutput=1

\usepackage[affil-it]{authblk}
\usepackage[usenames,dvipsnames]{xcolor}
\usepackage{amsfonts}
\usepackage{amsmath,amsthm,amssymb,dsfont}
\usepackage{enumitem}
\usepackage[english]{babel}
\usepackage{graphicx}	
\usepackage{float}
\usepackage[caption=false]{subfig}
\usepackage[margin=1in]{geometry}
\usepackage{url}
\usepackage{todonotes}
\usepackage{bbm}
\usepackage{booktabs}
\usepackage{scalefnt}
\usepackage[linesnumbered,ruled,vlined]{algorithm2e}

\definecolor{myyellow}{RGB}{242,226,149}
\definecolor{grey}{RGB}{150,150,150}
\definecolor{myblue}{RGB}{200,220,230}

\usepackage{pgfplots}
\usepgfplotslibrary{external}
\pgfplotsset{compat=newest}

\usepackage{mathrsfs}

\usepackage{hyperref}
\hypersetup{colorlinks=true,citecolor=blue,linkcolor=blue,filecolor=blue,urlcolor=blue,breaklinks=true}

\usepackage{nicefrac}
\usepackage{mathtools}

\theoremstyle{plain}
\newtheorem{theorem}{Theorem}[section]
\newtheorem{lemma}[theorem]{Lemma}

\newtheorem{proposition}[theorem]{Proposition}

\theoremstyle{definition}

\newtheorem{remark}[theorem]{Remark}

\newcommand*{\cC}{\mathcal{C}}
\newcommand*{\cD}{\mathcal{D}}

\newcommand*{\cH}{\mathcal{H}}

\newcommand*{\cN}{\mathcal{N}}

\newcommand*{\Z}{\mathbb{Z}}

\newcommand*{\CC}{\mathbb{C}}

\newcommand*{\<}{\left\langle}
\renewcommand*{\>}{\right\rangle}

\renewcommand{\hat}[1]{\widehat{#1}}

\newcommand{\dens}{\mathcal{D}}

\DeclareMathOperator{\tr}{Tr}
\DeclareMathOperator{\Tr}{Tr}
\newcommand{\ket}[1]{|#1\rangle}
\newcommand{\bra}[1]{\langle #1|}
\newcommand{\ketbra}[2]{|#1\rangle\langle#2|}
\newcommand{\id}{\mathbbm 1}

\newcommand{\edit}[1]{{\color{red}#1}}

\newcommand{\Loc}[1]{\hat{\cC}^{\mathrm{loc}}_{#1}}
\newcommand{\LocE}{\hat{\cC}^{\mathrm{loc}}_{E}}

\newcommand{\LocSOS}[1]{\hat{\cC}^{\mathrm{sos}}_{#1}}

\newcommand{\LocTI}[1]{\hat{\cC}^{\mathrm{loc,TI}}_{#1}}
\newcommand{\WM}[1]{\hat{\cC}^{\mathrm{WM}}_{#1}}
\newcommand{\WME}{\hat{\cC}^{\mathrm{WM}}_{E}}
\newcommand{\WMTI}[1]{\hat{\cC}^{\mathrm{WM,TI}}_{#1}}
\newcommand{\MEDTI}[1]{\hat{\cC}^{\mathrm{MED,TI}}_{#1}}
\newcommand{\TI}{\cC^{\mathrm{TI}}}
\newcommand{\MED}{\mathrm{MED}}

\author[1]{Hamza Fawzi}
\author[2]{Omar Fawzi}
\author[1]{Samuel O. Scalet\footnote{sos25@cam.ac.uk}}
\affil[1]{Department of Applied Mathematics and Theoretical Physics, University of Cambridge, United Kingdom}
\affil[2]{Univ Lyon, Inria, ENS Lyon, UCBL, LIP, Lyon, France}

\title{Entropy Constraints for Ground Energy Optimization}
\renewcommand{\edit}[1]{#1}

\begin{document}
\maketitle
\begin{abstract}
We study the use of von Neumann entropy constraints for obtaining lower bounds on the ground energy of quantum many-body systems.
Known methods for obtaining certificates on the ground energy typically use consistency of local observables and are expressed as semidefinite programming relaxations. The local marginals defined by such a relaxation do not necessarily satisfy entropy inequalities that follow from the existence of a global state. Here, we propose to add such entropy constraints that lead to tighter convex relaxations for the ground energy problem.
We give analytical and numerical results illustrating the advantages of such entropy constraints. We also show limitations of the entropy constraints we construct: they are implied by doubling the number of sites in the relaxation and as a result they can at best lead to a quadratic improvement in terms of the matrix sizes of the variables.
We explain the relation to a method for approximating the free energy known as the Markov Entropy Decomposition method.
\end{abstract}

\section{Introduction}

A fundamental computational problem in quantum many-body theory is to compute the ground energy of local Hamiltonians. Consider a multipartite Hilbert space $\cH = \otimes_{v \in V} \CC^d$ with local dimension $d$, on a finite set of sites $V$. A $k$-local Hamiltonian, is a Hermitian operator on $\cH$ defined as
\begin{equation}
\label{eq:localH}
H=\sum_{A\in\binom{V}{k}} h_A,
\end{equation}
where each of the $h_A$ is a Hermitian operator acting nontrivially only on the set $A \subset V$, and $\binom{V}{k}$ is the set of subsets of $V$ of size $k$. In this paper, we will be mostly interested in 2-local Hamiltonians, where the interaction can be modeled by a graph $G=(V,E)$ on the set of sites $V$, and where a Hamiltonian term $h_{ij}$ is attached to each edge $ij \in E$:
\begin{equation}
\label{eq:2local}
H=\sum_{ij \in E} h_{ij}.
\end{equation}

The \emph{ground energy} of $H$ is defined as its smallest eigenvalue. Due to the special structure of $H$, its matrix representation is generally sparse and thus one can apply standard methods such as Lanczos iterations \cite{golub2013matrix} to compute its minimal eigenvalue. However, since the dimension of $H$ grows exponentially with $|V|$, this is only feasible for moderate values of $|V|$. It is of considerable theoretical and practical interest to find efficient algorithms, that scale polynomially in $|V|$, to approximate the ground energy of local Hamiltonians~\cite{gharibian2015quantum,varbench2023}.

The smallest eigenvalue of $H$ admits the following variational formulation:
\begin{equation}
\label{eq:variational}
\lambda_{\min}(H) = \min_{\psi \in \cH} \frac{\<\psi, H \psi\>}{\<\psi, \psi\>}.
\end{equation}
Variational methods posit a certain form for the state $\psi = \psi_{\theta}$, and find the value of the parameters $\theta$ that minimize the objective function of \eqref{eq:variational}. As such, these methods provide upper bounds on $\lambda_{\min}(H)$. A prominent example are tensor network states~\cite{Bridgeman2017,Cirac2021}, which have been extremely successful and in particular give provably efficient algorithms for gapped systems in one dimension \cite{White1992,Hastings2007,Landau2015}.

Another class of methods that have been studied in the literature are based on convex relaxations and provide lower bounds on $\lambda_{\min}(H)$. For 2-local Hamiltonians $H$ of the form \eqref{eq:2local}, computing the energy $\<\psi, H \psi\>$ only requires knowledge of the two-body marginals $\rho_{ij}$ of $\ketbra{\psi}{\psi}$ for $ij \in E$. If we denote by $\cC$ the set of two-body marginals that are consistent with a global state on $V$, i.e.,
\begin{equation}
\label{eq:C}
\cC = \cC_{d,V,E} = \left\{ (\rho_{ij})_{ij \in E} : \exists \rho \in \dens(\otimes_{v \in V}\CC^{d}), \text{ s.t. } \rho_{ij} = \tr_{V\setminus \{i,j\}} \rho \right\}
\end{equation}
where $\cD(\cH)$ denotes the set of density operators on a Hilbert space $\cH$, 
then one can write the ground energy problem for a 2-local Hamiltonian \eqref{eq:2local} as a linear optimization problem over $\cC$:
\begin{equation}
\label{eq:lambdaminC}
\lambda_{\min}(H) = \min_{(\rho_{ij}) \in \cC} \;\; \sum_{ij \in E} \tr[h_{ij} \rho_{ij}].
\end{equation}
To make this approach tractable, it is required to have a computationally efficient representation of the convex set $\cC$. Unfortunately, it is highly likely that $\cC$ does not have any simple representation, e.g., it is known that the problem of checking membership in $\cC_{2,[n],\binom{[n]}{2}}$ is QMA-hard \cite{Liu2006}.

Rather than aiming to describe $\cC$, we are interested in constructing efficient outer relaxations of $\cC$, i.e., tractable convex sets $\hat{\cC}$ such that $\cC \subset \hat{\cC}$. Replacing $\cC$ by $\hat{\cC}$ in \eqref{eq:lambdaminC} would then yield a lower bound on $\lambda_{\min}(H)$. Such relaxations $\hat{\cC}$ can be constructed by identifying \emph{necessary conditions} that any set of marginals $(\rho_{ij})$ which are globally consistent must satisfy. Most relaxations that have been constructed in the literature are based on semidefinite programming. We describe here the most popular approaches:
\begin{itemize}
\item A simple relaxation can be obtained by simply imposing that the two-body marginals are consistent on the intersection of their supports, i.e., one can take
\begin{equation}
\label{eq:2bodyloccons}
\begin{aligned}
\LocE=\Bigl\{ (\rho_{ij})_{ij \in E} : & \;  \rho_{ij} \geq 0, \tr \rho_{ij}=1 \;\; \forall ij \in E\\
& \; \text{ and } \tr_{j} \rho_{ij} = \tr_{j'} \rho_{ij'} \;\; \forall ij,ij' \in E \Bigr\}.
\end{aligned}
\end{equation}
This relaxation can be made tighter by introducing higher-order marginals of $\rho$, namely one can consider
\begin{equation}
\label{eq:kbodyloccons}
\begin{aligned}
\Loc{l} = \Bigl\{ (\rho_{ij})_{ij \in E} : & \; \exists (\rho_{S})_{|S| \leq l}, \;\; \rho_{S} \geq 0, \tr \rho_{S}=1 \; \forall S \in \binom{V}{l}\\
& \; \text{ and } \tr_{S\setminus S'} \rho_{S} = \tr_{S' \setminus S} \rho_{S'} \;\; \forall S,S' \in \binom{V}{l} \Bigr\}.
\end{aligned}
\end{equation}
It is clear that $\cC = \Loc{N} \subset \Loc{N-1} \subset \dots \subset \Loc{2}\subset \Loc{E}$, where $N=|V|$.
\item The Lasserre/sum-of-squares relaxation \cite{Mazziotti2006,Pironio2010,Baumgratz2012, Barthel2012,brandao2013product} stems from the observation that if $\ket{\psi}$ is a global state on $\cH$, then $\<\psi, O^{\dagger} O \psi\> \geq 0$ for any observable $O$ acting on $\cH$. In particular if $O$ is a $l$-local operator, then $O^{\dagger} O$ is at most $2l$-local, and $\<\psi, O^{\dagger} O \psi\>$ is linear in the expectation values $m_F = \<\psi, F \psi\>$ of $2l$-local observables $F$. It turns out that the infinite family of constraints
\begin{equation}
\label{eq:sos-psd}
\<\psi, O^{\dagger} O \psi\> \geq 0 \;\; \forall O \text{ $l$-local observable on $\cH$}    
\end{equation}
can be encoded as a single positive semidefinite (psd) constraint of a matrix whose entries are linear in the expectation values $(m_F)$.
\edit{More precisely the matrix is indexed by strings of product operators, each taken from an operator basis, up to weight $l$.
For two such operators $O$, $P$, the $(O,P)$ entry of the matrix is then given by the expectation value of the product $O^\dagger P$.
Requiring this matrix to be psd corresponds to the constraint~\eqref{eq:sos-psd}.}
The corresponding relaxation $\LocSOS{l}$ can then be expressed as
\[
\LocSOS{l} = \left\{ (\rho_{ij})_{ij \in E} : \Tr[\rho_{ij}]=1, \rho_{ij} = \sum_{\alpha,\beta} m_{\sigma_i^{\alpha} \sigma_j^{\beta}} \sigma^{\alpha} \otimes \sigma^{\beta} \text{ and psd constraint \eqref{eq:sos-psd}} \right\}
\]
where $\sigma^{\alpha}$ is an orthonormal basis of operators acting on $\CC^d$.
\edit{Since it allows more flexible choices of operators to be included in the basis, the Lasserre hierarchy is a customary choice for settings without an underlying lattice geometry such as problems from quantum chemistry.}
\end{itemize}

\paragraph{Main contributions} In this paper, we are interested in relaxations for the convex set $\cC$ that go beyond semidefinite programming. In particular, we are interested in relaxations that use \emph{entropies} of the local marginals. Recall that the von Neumann entropy of the system $A$ for the state $\rho$ is defined by
\[
S(A)_{\rho} = -\tr \rho_{A} \log \rho_{A},
\]
where $\rho_{A}$ denotes the reduced state of $\rho$ on the system $A$.
In addition, the conditional entropy is defined by
\[
S(A|B)_{\rho} = S(AB)_{\rho} - S(B)_{\rho}.
\]
An important property about the latter is that it is concave in $\rho_{AB}$; this follows from the identity
\[
S(A|B)_{\rho} = -D(\rho_{AB} \| \id_A \otimes \rho_B)
\]
where $D(\rho \| \sigma) = \tr [\rho(\log \rho - \log \sigma)]$ is the relative entropy function, which is jointly convex in $(\rho,\sigma)$ \cite{Lindblad1974}. 

The main contribution of this paper is to study two families of entropy constraints that yield new strengthened relaxations for the ground energy problem. These relaxations are obtained by imposing inequalities on the conditional entropies of local marginals of the global state, and can be combined with any of the existing semidefinite relaxations. By virtue of the concavity of the conditional entropy function, these relaxations are all convex and can be solved efficiently using tools from convex optimization \cite{nesterov1994nemirovski,fawzi2022optimal}. 

The first family of entropy constraints come from weak monotonicity, see Section~\ref{sec:weakMono}, and the second family of constraints, that we call Markov Entropy Decomposition (MED) constraints, are motivated by the work of Poulin and Hastings~\cite{Poulin2011} and are based on combining the chain rule together with strong subadditivity, see Section~\ref{sec:med}. In Section~\ref{sec:inf}, we consider both families in the special case of infinite lattice systems. Though weak monotonocity constraints are in many cases stronger than MED constraints, we show that in general the two families are not comparable. Our main message is that for many natural Hamiltonians, imposing entropy constraints can lead to significantly tighter bounds compared to simple consistency conditions captured by $\Loc{l}$. Families of examples are presented in Figure~\ref{fig:3siteproj1} and detailed in Section~\ref{sec:quant-advantage} and then the Heisenberg XXZ-chain is discussed in Section \ref{sec:numerical}. We also show limitations on the gains that can be obtained using weak monotonicity constraints (and also MED in many settings): entropy constraints involving $l$ sites are implied by consistency constraints on $2l-1$ sites, see Eq~\eqref{eq:sandwichWM}. As a result, as the size of the matrix variables involved in $\Loc{l}$ is exponential in $l$, entropy constraints can at most lead to a quadratic improvement in terms of the size of the matrix variables.

We note that weak monotonicity has been mentioned previously as a necessary condition for global consistency in \cite{carlen2013}. Another related work considered entropic constraints to the number of orthogonal pure state extensions in a quantum marginal problem \cite{osborne2008}.

An intriguing open question arising from our work is to construct other entropy constraints that could lead to tighter relaxations. This question is related to obtaining inequalities for the so-called quantum entropy cone (see e.g.,~\cite{linden2005new}) though it differs in several respects: in our case, the dimension of the subsystems is fixed, the number of systems involved is bounded by $l$, and in order to obtain convex relaxations, we look for expressions that are concave in the state $\rho$, e.g., conical linear combinations of conditional entropies.

\section{Weak Monotonicity Constraints}\label{sec:weakMono}

We start by recalling the following well-known entropy inequality, also known as weak monotonicity.
\begin{lemma}[Weak monotonicity]\label{lem:weakMono}
For any state $\rho_{ABC}$ on systems $ABC$, we have
\begin{equation}\label{eq:weakMono}
S(A|B)_{\rho}+S(A|C)_{\rho}\ge0.
\end{equation}
\end{lemma}
\begin{proof}
Let $\rho_{ABCD}$ be a purification of $\rho_{ABC}$, i.e., $\rho_{ABCD}$ is rank-one and $\tr_{D}(\rho_{ABCD}) = \rho_{ABC}$. We have, using the fact that $\rho_{ABCD}$ is pure,
\begin{align*}
S(A|B)_{\rho} 
= S(AB)_{\rho} - S(B)_{\rho} 
= S(CD)_{\rho} - S(ACD)_{\rho} 
= - S(A|CD)_{\rho} 
\ge-S(A|C)_{\rho},
\end{align*}
where we used strong subadditivity, i.e., the property that $S(A|C) \geq S(A|CD)$.
\end{proof}
There are two important features of this inequality that we want to highlight: the first one is that it does not involve the global state of $\rho$ on $ABC$, but only the marginals of $\rho$ on $AB$ and $AC$. The second important aspect is that the inequality \eqref{eq:weakMono} defines a convex region in the space of marginals $(\rho_{AB},\rho_{AC})$. This is a consequence of the concavity of the conditional entropy function. 
For these reasons, the inequality can be used to strengthen the semidefinite relaxations $\hat{\cC}$ defined earlier, as follows. For example, the set of 2-body locally consistent marginals $\Loc{2}$ can be strengthened by adding the following scalar inequalities:
\begin{equation}
    \label{eq:wm2body}
    S(i|j)_{\rho} + S(i|k)_{\rho} \geq 0 \quad \forall 1\leq i,j,k\leq n \text{ distinct}.
\end{equation}

In general, if one considers relaxations involving marginals on $l \geq 2$ sites, one can include all weak monotonicity inequalities \eqref{eq:weakMono} for all disjoint sets $A,B,C \subset V$ such that $|AB| \leq l$ and $|AC| \leq l$. The next lemma shows that it is indeed sufficient to consider only inequalities where $|A|=1$, and $|B|=|C|=l-1$, i.e.,
\begin{equation}
\begin{aligned}
\WM{l} = \Bigl\{ (\rho_{ij})_{ij \in E} : & \; \exists (\rho_{S})_{|S| \leq l}, \;\; \rho_{S} \geq 0, \tr \rho_{S}=1 \; \forall S \in \binom{V}{l}\\
& \; \text{ and } \tr_{S\setminus S'} \rho_{S} = \tr_{S' \setminus S} \rho_{S'} \;\; \forall S,S' \in \binom{V}{l} \\
&\; \text{ and } S(A|B)+S(A|C)\ge0\;\forall |A| = 1, B,C\in\binom{V\setminus A}{l-1} \text{ disjoint}\Bigr\}.
\end{aligned}
\end{equation}

\begin{lemma}
Let $|V|\ge 2l-1$ and $(\rho_S)_{S \in \binom{V}{l}}$ be a set of locally consistent marginals satisfying weak monotonicity $S(A|B)_{\rho}+S(A|C)_{\rho}\ge0$ for any disjoint $A$, $B$ and $C$ with $|A|=1$ and $|B|=|C|=l-1$. Then the same weak monotonicity inequality holds for any disjoint sets $A$, $B$ and $C$ of any size (as long as it is defined, i.e., $|AB|,|AC|\le l$).
\end{lemma}
\begin{proof}
For $|A|=1$ and $B, C$ of any size at most $l-1$, the proof directly follows from the data-processing inequality for conditional entropy.
If $|A| = m > 1$, let the elements of $A$ be $A_1 \dots A_m$ with $A_i \in V$ and use the chain rule to \edit{expand $S(A|B)$ and $S(A|C)$ in a telescoping sum that cancels all but the first and last terms}
\begin{align*}
S(A|B)_{\rho} + S(A|C)_{\rho} 
&= \sum_{i=1}^m S(A_i | A_1 \dots A_{i-1} B)_{\rho} + \sum_{i=m}^1 S(A_i | A_{i+1} \dots A_{m} C)_{\rho} \\
&= \sum_{i=1}^m S(A_i | A_1 \dots A_{i-1} B)_{\rho} + S(A_i | A_{i+1} \dots A_{m} C)_{\rho}.
\end{align*}
Using the fact that $|AB| \leq l$, we have that $|A_1 \dots A_{i-1} B| \leq l-1$ for any $i \in \{1, \dots, m\}$ and similarly for $C$, together with the fact that $A_1 \dots A_{i-1} B$ and $A_{i+1} \dots A_{m} C$ are disjoint we obtain that $S(A|B)_{\rho} + S(A|C)_{\rho} \geq 0$.
\end{proof}
\edit{We remark that the Lemma still holds if the marginal relaxation is not defined for \emph{all} $l$-site subsets, as long as it contains the marginals and entropy inequalities that appear in the proof.}

Notice that if $\rho_{AB}$ and $\rho_{AC}$ are classical probability distributions, or more generally if $\rho_{AB}$ and $\rho_{AC}$ are separable quantum states, then conditional entropies are nonnegative and thereby Eq.~\eqref{eq:weakMono} is automatically satisfied, even without assuming the existence of a global state $\rho_{ABC}$. As a result, imposing the weak monotonicity inequality is useless for classical Hamiltonians. Furthermore, we note that the weak monotonicity inequalities at level $l$ are automatically implied by the level $2l-1$ local consistency relaxation, i.e.,
\begin{align}
\label{eq:sandwichWM}
\Loc{2l-1} \subset \WM{l} \subset \Loc{l}.
\end{align}
This is simply because the level $2l-1$ ensures that the marginals are consistent with a valid quantum state on $ABC$, for which weak monotonicity is known to hold.

While this means weak monotonicity constraints cannot help more than doubling the number of sites in the marginal relaxation, it should be noted that this corresponds to squaring the size of each variable in the semidefinite program.

\paragraph{Graphical illustration} We end this section by graphically illustrating the effect of the entropy constraints. Consider the problem of characterizing 
 the $\{1,2\}$ and $\{2,3\}$-body marginals of a 3-site qubit state $\rho = \rho_{123}$, i.e.,
\[
\left\{ (\rho_{12},\rho_{23}) : \rho_{ij} = \tr_{\setminus \{i,j\}} \rho_{123} \text{ for some } \rho_{123} \in \dens((\CC^2)^{\otimes 3}) \right\}.
\]
Figure \ref{fig:3siteproj1} shows a particular two-dimensional projection of this set on Bell states, namely 
\[
(\rho_{12},\rho_{23})\mapsto \left( \tr[(\rho_{12}+\rho_{23})/2 \ket{\psi^{-}} \bra{\psi^{-}}] , \tr[(\rho_{12}+\rho_{23})/2 \ket{\phi^{+}} \bra{\phi^{+}}] \right).\]
\begin{figure}[ht]
\centering
\input{starSet.tex}
\caption{\label{fig:3siteproj1}We compare the sets of valid states on 3 sites (blue, inner curved line), 2-body marginal relaxations (black triangle), and the entropy constrained marginals fulfilling $S(2|1) + S(2|3) \geq 0$ (red, outer curved line). We also depict the linear inequality mentioned in Eq.~\eqref{eq:starBound} (orange, dashed).
}
\end{figure}

\subsection{Quantitative advantage}
\label{sec:quant-advantage}

In this section we analyze quantitatively the advantage that entropy constraints can provide for estimating the ground energy of certain local Hamiltonians.

\paragraph{A Hamiltonian on 3 sites} We start by looking at a Hamiltonian defined on a 3-node graph given by
\[
H = h_{12} + h_{23},
\]
where
\[
h_{12} = -\frac12 \ketbra{\psi^-}{\psi^-}_{12} \qquad h_{23} = -\frac12 \ketbra{\psi^-}{\psi^-}_{23}
\]
and $\ket{\psi^-}=(\ket{01}-\ket{10})/\sqrt2$. The ground energy of this Hamiltonian is $\lambda_{\min}(H) = -3/4$. Using the relaxation $\Loc{2}$, we get the (trivial) lower bound $-1$: indeed a valid locally consistent assignment of 2-body marginals is 
\begin{equation}
    \label{eq:rhoijbadassignment}
    \rho_{12}=\rho_{23}=\ketbra{\psi^-}{\psi^-}
\end{equation}
as all 1-body marginals are consistently equal to the maximally mixed state $\id/2$. However, this violates the monogamy of entanglement property as for a valid global state, $2$ cannot be maximally entangled with $1$ and $3$.
This violation of monogamy is to some extent captured by entropy constraints coming from weak monotonicity. Indeed, the assignment \eqref{eq:rhoijbadassignment} violates the inequality $S(2|1)_{\rho}+S(2|3)_{\rho}\ge0$ since the left-hand side in this case is equal to $-2$. In fact, by optimizing the Hamiltonian $H$ over the relaxation $\WM{2}$ we get a value that is approximately $-0.811 > -1$, 
i.e.,
\begin{equation}\label{eq:singleEdgeOpt}
\min_{(\rho_{ij}) \in \WM{2}} \tr[h_{12} \rho_{12}] + \tr[h_{23} \rho_{23}] \geq -0.811.
\end{equation}
Graphically this value can be seen in Figure \ref{fig:3siteproj1} as the $x$-component of the rightmost point of the red convex set, whereas the true value $-0.75$ is for the blue convex set.

\begin{remark}[Exploiting symmetries] In the above example, one can actually use the symmetries of the Hamiltonian to simplify the entropy constraints.
As the Hamiltonian commutes with the unitary operation of exchanging subsystems 1 and 3, the minimizer can be chosen to obey the same symmetry.
For that state we have $S(2|1)=S(2|3)$, and so the weak monotonicity inequality becomes simply $S(2|1) \geq 0$. We will see more examples of similar arguments when we consider translation-invariant Hamiltonians in Section \ref{sec:inf}.
\end{remark}

\paragraph{Larger graphs}
\edit{We now show that the advantage of entropy constraints extends to Hamiltonians on arbitrary graphs.} Consider a Hamiltonian defined on a general connected graph $G=(V,E)$ of the form
\begin{equation}
    \label{eq:Hepr}
    H = \sum_{ij \in E} h_{ij} \qquad h_{ij} = -\ketbra{\psi^-}{\psi^-}_{ij} \;\; \forall ij \in E.
\end{equation}
While the simple relaxation $\Loc{2}$ yields a trivial lower bound of $-|E|$ on $\lambda_{\min}(H)$, one can show that the relaxation incorporating the weak monotonicity constraints $S(i|j) + S(i|k) \geq 0$ will have a value $\gtrsim -0.811 |E|$. More precisely, we prove:
\begin{theorem}
\label{thm:EPRcst}
For the Hamiltonian \eqref{eq:Hepr} defined on a connected graph $G=(V,E)$, we have
\[
-|E|=\min_{(\rho_{e})\in \LocE }\sum_{e\in E} \Tr[h_e\rho_e]\le-0.811|E-1|-1\le\min_{(\rho_{e})\in \WME }\sum_{e\in E}\Tr[h_e\rho_e].
\]
\end{theorem}
The key to proving this theorem is the following proposition, which allows us to decompose the edges $E$ of the graph into disjoint pairs of adjacent edges.
\begin{proposition}\label{prop:adjEdge}\sloppy
For any connected graph $G=(V,E)$, there exists a set $P = \{\{f_1,g_1\},\ldots, \{f_m,g_m\}\}$ of disjoint pairs of adjacent edges (i.e., $f_i \in E$ and $g_i \in E$ share a node for each $i$, and $\{f_i,g_i\}\cap \{f_j,g_j\} = \emptyset$ for all $i\neq j$) that covers all edges of the graph for an even number of edges and all but one in the case of an odd number of edges.
\end{proposition}
\begin{proof}
    See Appendix \ref{sec:adjEdge}.
\end{proof}
\begin{proof}[Proof of Theorem \ref{thm:EPRcst}]
We use the decomposition of the edge set $P$ from Proposition \ref{prop:adjEdge} and write
\[
H=\sum_{\{f,g\}\in P} (h_f+h_g)+\delta_{odd} h_e
\]
where $\delta_{odd}=1$ if the graph has an odd number of edges with $e$ the unmatched edge and $\delta_{odd}=0$ otherwise.
We estimate
\begin{align*}
\min_{(\rho_{e})\in \WME}\sum_{e\in E}\Tr[h_e \rho_e]
&\ge\left(\sum_{\{f,g\}\in P} \min_{(\rho_{e})\in \WME}\Tr[h_f\rho_f]+\Tr[h_g\rho_g]\right)+\delta_{odd}\\
&\ge-0.811 (|E|-\delta_{odd})-\delta_{odd}
\end{align*}
where we bounded for each term the objective by the optimal value for the entropy constrained value of the 3-node graph from Eq.~\eqref{eq:singleEdgeOpt}.
\end{proof}

We remark that for this problem, we chose a smaller set of marginals and entropy constraints motivated by the structure of the problem rather than simply all possible constraints for the given size of marginals.

These examples are unweighted instances of the quantum version of the Max-Cut problem that has been widely studied in the literature \cite{Gharibian2019, Parekh2021, Anshu2020, Parekh2020, King2022, lee2022optimizing}.
These works consider semidefinite relaxations of this problem and then round the solutions of the semidefinite program to a valid global quantum state. In order to do this, for example in~\cite{Parekh2021}, it is shown that any $(\rho_{12}, \rho_{23})$ coming from the relaxation  $\LocSOS{2}$ satisfies the linear inequality
\begin{equation}\label{eq:starBound}
\frac{1}{2}\left(\Tr[\rho_{12} \ketbra{\psi^-}{\psi^-}_{12}]+\tr[\rho_{23}\ketbra{\psi^-}{\psi^-}_{23}]\right)\le0.75.
\end{equation}
This can be interpreted in terms of Figure~\ref{fig:3siteproj1}: it shows that Eq.~\eqref{eq:starBound} leads to a tight bound when optimizing the linear function defined by $\ketbra{\psi^-}{\psi^-}$ on sites $12$ and $23$. The entropy constraint performs worse in the directions of the Bell states as seen in Figure~\ref{fig:3siteproj1} and in Equation~\eqref{eq:singleEdgeOpt}, but it is better in some other directions.

\paragraph{Larger marginals} We extend the above example to show how nontrivial entropy constraints can be easily constructed for any choice $k \geq 2$ of number of sites considered by the marginal.
Let us consider a $k$-local Hamiltonian, where each term is given by a projector on the GHZ-state $h_e=-\ketbra{\textup{GHZ}}{\textup{GHZ}}$, which is given by $\ket{\textup{GHZ}}=(\ket{0\ldots0}+\ket{1\ldots1})/\sqrt{2}$ and $e$ is a subset of the sites of size $k$.
Again, for the $k$-site relaxation the optimal solution is to assign a projector onto the GHZ-state to every marginal $\rho_e=\ketbra{\textup{GHZ}}{\textup{GHZ}}$, which are all locally consistent.
However, the conditional entropy for the GHZ state $S(1\ldots i|i+1\ldots k)=-1$ is negative. As a result, for two distinct subsets $e$ and $f$ of size $k$ that have a nonempty intersection, if we write $e = AB$ and $f = AC$, the weak monotonicity constraint $S(A|B)+S(A|C)\ge0$ is violated.
For concreteness, let $V=1\ldots (2l-1)$ and consider hyperedges $e=1\ldots l,\;f=l\ldots (2l-1)$ and consider the Hamiltonian $h_e + h_f$. For the relaxation we specify the marginals $\rho_e$, $\rho_f$ subject to local consistency on site $l$ and the weak monotonicity constraint
\[
S(l|1\ldots l-1)+S(l|l+1\ldots 2l-1)\ge0
\]
In fact, due to the symmetry of the problem, this can be reduced to the optimization
\begin{equation}\label{eq:hyperNum}
\min_{\rho_e,\;S(l|1\ldots l-1)\ge0}\Tr[h_e \rho_e].
\end{equation}
For any $l$, this problem is related to the problem at $l=2$
by the local isometry $U_{1\to 1\ldots l-1}=\ketbra{0\ldots0}{0}+\ketbra{1\ldots1}{1}$. Since this local isometry leaves the conditional entropy invariant the optimal value of Eq.~\eqref{eq:hyperNum} does not depend on $l$ and is thereby equal to approximately $-0.811$ as in the previously encountered example Eq.~\eqref{eq:singleEdgeOpt}.

\edit{
\paragraph{Lasserre/SoS-hierarchy} We omit a detailed study of the enhancement of the Lasserre hierarchy mentioned in the introduction via entropy constraints, but comment on some previous results and future directions.
The Lasserre hierarchy at level $l$ provides expectation values on up to $2l$ sites, but only guarantees that they are consistent with a density matrix for up to $l$ sites.
Thereby, it is not possible to define any entropy constraints (involving at least 2 sites) at the first level.
To circumvent the lack of positivity, in~\cite{Parekh2021} a level 1.5 of the Lasserre hierarchy is introduced that strengthens the level 1 by an explicit positivity constraint on the 2-site marginals.
In fact, it is used for the analysis of star graphs including the one in this section, for which it yields an optimal value of 0.862.
Since this is a worse approximation compared to the entropy constrained value we obtain this shows that there is room for improvement using the entropy constraint, which can be defined for the Lasserre 1.5.
It is a future direction to find explicit improvements also at higher levels of the Lasserre hierarchy.
They can be imposed whenever the constraints of the hierarchy enforce the positivity of the marginals involved in the entropy constraints.
}

\section{Markov Entropy Decomposition}\label{sec:med}

In this section we derive another family of entropy constraints that one can impose on the set of globally consistent marginals. We call these inequalities \emph{Markov Entropy Decomposition (MED)} inequalities, because they have appeared in \cite{Poulin2011} in the context of deriving lower bounds on the free energy of a Hamiltonian $H$. The next lemma briefly describes the family of inequalities. We explain the connection to the free energy later.
\begin{lemma}[Markov Entropy Decomposition inequalities]
\label{lem:medIneq}
Let $\rho$ be a density operator acting on the Hilbert space corresponding to the sites $V$. Consider an ordering $1,\ldots,N=|V|$ of the sites in $V$, and for each $i \in \{1,\ldots,N\}$, let $\cN_i \subset \{1, \dots, i-1\}$ be a subset of the sites appearing before $i$ (in the chosen order) with $|\cN_i| \leq l-1$. Then for any $1\leq k \leq N$
\begin{equation}
\label{eq:medIneq}
\sum_{i=1}^k S(i|\cN_i)_{\rho} \geq 0.
\end{equation}
\end{lemma}
\begin{proof}
We have, using the chain rule and strong subadditivity:
\[
0 \leq S(1 \dots k)_{\rho} = S(1)_{\rho} + S(2|1)_{\rho} + S(3|12)_{\rho} + \dots + S(k|1\ldots k-1)_{\rho} \leq \sum_{i=1}^k S(i|\cN_i)_{\rho}.
\]
\end{proof}
We see that the inequalities \eqref{eq:medIneq} only involve the marginals of $\rho$ on at most $l$ sites (since $|\cN_i| \leq l-1$) and that they define a convex region in the space of these marginals. Thus these inequalities can be used to strengthen the relaxations $\Loc{l}$.  We note that there is an MED inequality for each choice of ordering of the sites in $V$, and each choice of \emph{Markov shields} $\cN_i \subset \{1,\ldots, i-1\}$.

\subsection{Relationship with weak monotonicity}

A natural question is whether the MED inequalities are equivalent to the weak monotonicity inequalities derived in the previous section. The answer in general is no, and we show in this section that the two relaxations are incomparable.

However, we first observe that some simple MED inequalities are implied by the weak monotonicity inequalities. For example consider the case where the Markov shields $\cN_i$ consist of the (at most) $l-1$ sites that are immediately preceding the site $i$ in the chosen order, i.e., $\cN_i = \cN_i^{-} = \{\max(1,i-l+1),\ldots,i-1\}$. Then one can show that the MED inequality is implied by the weak monotonicity inequalities. Indeed, let $\cN_i^{+}$ be the (at most) $l-1$ sites that are immediately \emph{following} the site $i$ in the chosen order, i.e., $\cN_i^+ = \{i+1,\ldots,\min(k,i+l-1)\}$. The MED inequality states
\begin{equation}
\label{eq:medorderline}
\sum_{i=1}^k S(i|\cN_i^{-}) \geq 0.    
\end{equation}
One can check that $\sum_{i=1}^k S(i|\cN_i^{-}) = \sum_{i=1}^k S(i|\cN_i^+)$, and so the inequality above can be equivalently written as
\[
\sum_{i=1}^k S(i|\cN_i^{-}) + S(i|\cN_i^+) \geq 0.
\]
Now it suffices to observe that for each $i$, $S(i|\cN_i^{-}) + S(i|\cN_i^+) \geq 0$ is implied by weak monotonicity.

We now give two examples showing that MED and weak monotonicity constraints are in general inequivalent.

\paragraph{Weak monotonicity does not imply MED}

We construct a set of 2-body locally consistent states on $N$ sites (for $N$ large enough) that violate a certain MED inequality but are consistent with all possible weak monotonicity constraints of the form \eqref{eq:wm2body}.
The 2-body marginals are defined as follows (for $i,j \geq 2$):
\begin{align*}
\rho_{1i}&=\lambda |\psi^-\rangle\langle\psi^-|+(1-\lambda)(\id/2\otimes|0\rangle\langle0|)\\
\rho_{ij}&=(\lambda\id/2+(1-\lambda)\ketbra{0}{0})^{\otimes2}.
\end{align*}
Here, $\ket{\psi^-}=(\ket{01}-\ket{10})/\sqrt2$ as before. The entropies of the single-body and two-body marginals are
\[
S(1) = 1 \quad \text{and} \quad S(i) = h(1-\frac\lambda2) \quad \text{and} \quad S(ij) = 2h(1-\frac\lambda2) \quad i,j \geq 2
\]
where $h$ is the binary entropy function, which is concave, continuous and satisfies $h(1)=0$ and $h(1/2)=1$.
We omit an explicit formula for the entropy of the two-body marginal $S(1i)$, but it is also concave, continuous, and equal to 1 if $\lambda=0$, and 0 if $\lambda=1$.

This means we can pick $0<\lambda<1$ such that $S(1|i)+S(1|j)=2 S(1i)-2h(1-\lambda/2)=0$.
Furthermore,
\begin{align*}
S(i|j)+S(i|1)&=h(1-\lambda/2)+S(i1)-1\\
&\ge \lambda + (1-\lambda)-1=0,
\end{align*}
where we applied concavity of $h$ and the entropy function.
Finally, $S(i|j)+S(i|k)\ge0$ for $i,j,k \geq 2$, because both marginals are product states, so all weak monotonicity constraints that can be defined for 2-site marginals are satisfied.

However, let us consider the MED inequality
\[
0 \leq S(1)+\sum_{i=1}^N S(i|1)=1+N (S(1i)-1).
\]
Since $0<S(1i)=h(1-\lambda/2)<1$ ($h$ is strictly decreasing on $[1/2,1]$), we see that the right-hand side becomes negative for sufficiently large $N$ ($N=8$ is sufficient).

\paragraph{MED does not imply weak monotonicity}
We now construct a set of 2-body locally consistent states on $N=3$ sites that satisfy all the MED inequalities but violate some weak monotonicity inequalities.
We choose for all $ij$
\begin{equation}
\label{eq:3siteMEDvsWM}
\rho_{ij}=\lambda\ketbra{\psi^-}{\psi^-}+(1-\lambda)\id/4    
\end{equation}
with $\lambda$ such that $-1/2\le S(i|j)<0$. 
This is possible due to the continuity in $\lambda$ and the fact that it takes the values $1$ and $-1$ for $\lambda=0$ and $\lambda=1$ respectively. 
This choice violates weak monotonicity as all conditional entropies are negative.
Due to the permutation invariance, the only relevant MED constraints to be checked are
\begin{align*}
S(1)+S(2|1)+S(3|2)&\ge0\\
S(1)+S(2|1)+S(3|1)&\ge0.
\end{align*}
These are satisfied since $S(1) = 1$ and each $S(i|j) \geq -1/2$.

\begin{remark}
We note that on three sites, the 2-body MED constraints are all implied by weak monotonicity. Indeed, the MED constraints are all of the form
\[
S(i)+S(j|i) + S(k|j) \geq 0
\]
for all choices of $i,j,k$ distinct. The only other choice $S(i)+S(j|i)+S(k|i)=S(j)+S(i|j)+S(k|i)$ is in fact equivalent. This is precisely of the form \eqref{eq:medorderline} where the sites are ordered as $i<j<k$, which have been shown to be implied by the weak monotonicity constraints. The above example \eqref{eq:3siteMEDvsWM} thus shows that the weak monotonicity constraints are strictly stronger than the MED constraints on 3 sites.
\end{remark}

\subsection{Connection with the free energy}

The expression appearing in the MED inequalities was used in \cite{Poulin2011} to derive lower bounds on the free energy $F(T)$ of the Hamiltonian $H$ on sites $V$ at temperature $T$. Recall that the free energy is given by the following variational formula
\[
F(T) = \min_{\rho} \{ \Tr[H\rho]-T S(V)_{\rho} \},
\]
where the minimization is over all density operators $\rho$ acting on the Hilbert space of all the sites in $V$.
As opposed to the ground energy problem, the objective does not only depend on few-site marginals as the entropy function involves the global density matrix and its spectrum. The entropy however can be decomposed into conditional entropies which can be upper bounded using the strong subadditivity inequality:
Considering an order of the sites $1,\ldots,N$, we write
\begin{align*}
S(V)_{\rho}&=S(1)_{\rho}+S(2|1)_{\rho}+\ldots+S(N|1\ldots N-1)\\
&\le \sum_{i=1}^N S(i|\mathcal N_i)_{\rho}
\end{align*}
again using Markov shields $\mathcal N_i\subset \{1, \dots, i-1\}$ defining the sites that are taken into account in the conditioning.
The intuition in choosing these is that sites that are far away in the hypergraph distance induced by the interaction hypergraph have only small correlation as measured by the conditional mutual information (CMI), see for example \cite{Kato2019} for a result in 1D.
The CMI is also equal to the error made in the approximation using strong subadditivity.
This relaxation can now again be computed by just the marginals on $\{i\}\cup\mathcal N_i$ and the supports of the Hamiltonian terms and therefore allows for a further relaxation as explained in the previous section:
\begin{equation}
    \label{eq:MEDT}
\begin{aligned}
F(T) &=\min_{\rho} \{ \Tr[H\rho]-T S(\rho) \}\\
&\geq \min_{(\rho_{e}) \in \Loc{l}} \left\{ \sum_{e} \tr[h_{e} \rho_{e}] - T\sum_{i=1}^N S(i|\mathcal N_i) \right\} := \MED(T)
\end{aligned}
\end{equation}
Here, $l$ is chosen sufficiently large to define all conditional entropies and interaction terms.

\paragraph{Lower bounds for the ground energy}
It has been pointed out in \cite{Poulin2011} that while the free energy $F(T)$ is a decreasing function of $T$, its $\MED$ approximation $\MED(T)$ need not be.
At every temperature, however, we have $\MED(T)\le F(T)\le F(0)$, which is to say that the MED lower bounds the ground energy $F(0)$. We can maximize this lower bound over $T$ to obtain the best lower bound on the ground energy at that level. It is easy to see that maximizing the expression for $\MED(T)$ in \eqref{eq:MEDT} corresponds to adding the $\MED$ term as an inequality, i.e.,
\[
F(0) \geq \max_T \MED(T) = \min_{\substack{(\rho_{e}) \in \Loc{l}\\\sum_{i=1}^N S(i|\mathcal N_i) \geq 0}} \; \sum_{e} \tr[h_{e} \rho_{e}].
\]

\section{Infinite systems}\label{sec:inf}
A practically relevant case of the above methods are lattice systems. We start this section by discussing the application of weak monotonicity to one-dimensional translation-invariant systems.

\paragraph{One-dimensional systems} We consider a translation-invariant Hamiltonian $H$ acting on a chain of length $m$ of the form
\[
H_{[1,m]} = \sum_{i=1}^m h_{i,i+1}
\]
where the interaction terms $h_{i,i+1}$ act on sites $i$ and $i+1$, and are translates of the same Hermitian operator $h$ acting on two sites $\CC^d \otimes \CC^d$.
We use periodic boundary conditions, i.e., we identify $m+1$ with $1$.
We can define the ground energy per site of the infinite system as the limit 
\begin{equation}
\label{eq:e0h}
e_0(h) = \lim_{m\to\infty}\frac{1}{m} \lambda_{\min}(H_{[1,m]}) = \lim_{m\to\infty} \frac{1}{m} \min_{\rho_m \in \dens(\otimes_{v \in [1,m]} \CC^d)} \tr[H_{[1,m]} \rho_m].
\end{equation}
The existence of the limit above is shown in Appendix~\ref{sec:existence}; in fact we show that in the minimization problems in \eqref{eq:e0h}, it is equivalent to restrict to just a single interaction term while taking $\rho_m$ translation-invariant, i.e., satisfying $\tr_1[\rho_m] = \tr_m[\rho_m]$, since for such translation-invariant states the two-body marginals $\rho_{12}, \ldots, \rho_{m-1,m}$ are all equal. This allows us to express the ground energy density $e_0(h)$ in the following way
\[
e_0(h)=\min_{\rho\in\TI}\Tr[h\rho],
\]
where $\TI$ is the set of two-body density matrices that are extendible to an infinite translation-invariant system, i.e.,
\begin{equation}
\label{eq:C2TI}\TI = \{\rho \in \dens(\CC^d \otimes \CC^d) : \forall m\;\exists \rho_m\in\dens(\otimes_{v\in[1,m]}\CC^d),\;\rho=\Tr_{[3,m]}[\rho_m],\;\Tr_1[\rho_m]=\Tr_m[\rho_m]\}.
\end{equation}
A ``local'' relaxation of this set can be defined in a similar way as discussed previously, and leads to
\begin{equation}
\label{eq:loctil}
\LocTI{l}=\left\{ \rho \in \dens(\CC^d \otimes \CC^d) :\exists \rho_l\in\dens(\otimes_{v\in[1,l]}\mathbb C^d),\;\rho=\Tr_{[3,l]}[\rho_l],\;\Tr_1[\rho_l]=\Tr_l[\rho_l] \right\}.
\end{equation}
We prove in Appendix~\ref{sec:existence} that for any $h$, the value
\[
\min_{\rho \in \LocTI{l}} \tr[h\rho]
\]
converges to $e_0(h)$ as $l\to \infty$ at the rate $O(1/l)$.

Using weak monotonicity and translation-invariance, one can strengthen the relaxation \eqref{eq:loctil} by adding the inequality $S(l|1\ldots l-1) \geq 0$, leading to:
\begin{equation}
\begin{aligned}
\label{eq:WMTIl}
\WMTI{l}=\Big\{ \rho \in \dens(\CC^d \otimes \CC^d) :&\exists \rho_l\in\dens(\otimes_{v\in[1,l]}\mathbb C^d),\\
&\;\rho=\Tr_{[3,l]}[\rho_l],\;\Tr_1[\rho_l]=\Tr_l[\rho_l], \; S(l|1\ldots l-1)\ge0\Big\}.
\end{aligned}
\end{equation}
The next theorem shows that the above is a valid relaxation, and moreover that the resulting relaxation cannot be better than $\LocTI{2l-1}$.

\begin{theorem}
\[
\LocTI{2l-1}\subset \WMTI{l}\subset \LocTI{l}
\]
\end{theorem}
\begin{proof}The second inclusion is immediate.
To prove the first one we start from a translation-invariant state on sites $1,\ldots, 2l-1$.
We know from weak monotonicity that $S(l|1\ldots l-1)+S(l|l+1\ldots 2l-1)\ge0$.
However, translation-invariance tells us that
\begin{align*}
S(l|l+1\ldots 2l-1)&=S(l\ldots2l-1)-S(l+1\ldots 2l-1)\\
&=S(1\ldots l)-S(1\ldots l-1)\\
&=S(l|1\ldots l-1),
\end{align*}
so $2S(l|1\ldots l-1)\ge0$, which concludes the proof.
\end{proof}

In fact, the entropy constraint $S(l|1\ldots,l-1) \geq 0$ can as well be derived from the MED inequalities (Lemma \ref{lem:medIneq}) by using an increasing order and a Markov shield equal to the previous $l-1$ sites.

\paragraph{Higher dimensions}
We now consider a lattice system on $\Z^D$ for $D \geq 1$. Let $h$ be a Hamiltonian term on the origin and its nearest neighbours in positive axis directions (i.e., a Hermitian matrix of size $d^{D+1}\times d^{D+1}$), and for $v \in \Z^D$, let $h_v$ be its translate to the site $v \in \Z^D$, so that it acts on $\{v,v+e_1,\ldots,v+e_D\}$. Let us consider Hamiltonians of the form:
\[
H_{[1,m]^D} = \sum_{v \in [1,m]^{D}} h_v,
\]
with periodic boundary conditions, i.e., identifying $m+1$ with $1$. The ground energy density of the system is defined by the limit (see Appendix~\ref{sec:existence} for existence).
\[
e_0(h) = \lim_{m\to \infty} \frac{1}{m^D} \lambda_{\min}(H_{[1,m]^D}).
\]
Using the same considerations as in the 1D case, we can also write $e_0(h)$ as a linear optimization problem over the set of density matrices defined on $\{0,e_1,\ldots,e_D\}$ that are extendible to an infinite translation-invariant system, i.e.,
\[
\TI = \bigcap_{m\in [2,\infty)} \LocTI{[1,m]^D}
\]
where for any finite subset $A \subset \Z^D$ such that $\{0,e_1,\ldots,e_D\} \subset A$ we define 
\[
\begin{aligned}
\LocTI{A} = \Big\{ \rho \in \dens( (\CC^d)^{D+1} ) : &\exists \rho_A \in \dens( \otimes_{v \in A} \CC^d ) \text{ s.t. } \rho=\Tr_{A\setminus\{0,e_1,\ldots,e_D\}}[\rho_A]\\
&\Tr_{A\setminus A+t}[\rho_A] = \tr_{A+t\setminus A}[\tau_t(\rho_A)]\;\forall t\in\Z^D\Big\},
\end{aligned}
\]
where $\tau_t$ defines the translation operator shifting the state by $t$ sites.

We investigate what entropy constraints arising from weak monotonicity and the MED can be imposed to strengthen the relaxation $\LocTI{A}$ of $\TI$. The weak monotonicity inequality allows us to strengthen the relaxation above by adding constraints of the form
\begin{equation}
\label{eq:WMhighD}
S(i|B) + S(i|C) \geq 0
\end{equation}
for any choice of $i\in A$ and $B,C \subset A$ disjoint.

We now consider MED inequalities. 
Consider a translation-invariant state defined on $[-m,m]^D$ for large $m$. We fix an order of the sites $1,\ldots,N=(2m+1)^D$ that decreases in each coordinate direction and iterates over the coordinates (see Figure~\ref{fig:markovShield}). Let $\cN$ be some fixed region of size $l-1$ independent of $m$, which is a subset of the sites preceding the origin (a common choice would be the intersection of a ball $B(0,r)$ of fixed radius $r$ with the sites preceding the origin, see Figure \ref{fig:markovShield}).
We define the Markov shields $\mathcal N_i$  relative to the site $i$ by translating $\cN$ to the site numbered $i$, i.e., abusing notations:
\[
\cN_i = (i+\cN) \cap [-m,m]^D.
\]
\begin{figure}[H]
\centering
\input{markovShield}
\caption{\label{fig:markovShield}Markov shield with order of sites in 2D.}
\end{figure}
The MED inequality \eqref{eq:medIneq} in this case takes the form
\[
0 \leq \sum_{i=1}^N S(i|\cN_i).
\]
For sites $i$ that are not close to the boundary of the region $[-m,m]^D$, the regions $\{i\} \cup \cN_i$ are all translates of each other, and so the terms in the equation above are equal, by translation invariance. The number of sites $i$ that are close to the boundary is $o(N)$, and so dividing by $N$ and letting $N\to \infty$, the inequality above yields
\begin{equation}
\label{eq:MEDTIhighD}
0 \leq S(0|\cN)
\end{equation}
for a state in $\TI$.
This suggests the relaxation
\begin{align*}
   \TI\subset\MEDTI{0\cup\mathcal N} = \Big\{ \widetilde{\rho} \in \dens( (\CC^d)^{D+1} ) \ : \ &\exists \rho \in \dens( \otimes_{v \in 0\cup\mathcal N} \CC^d ) \text{ s.t. } 
    \\&\Tr_{(0\cup\cN)\setminus (0\cup\cN)+t}[\rho] = \tr_{(0\cup\cN)+t\setminus (0\cup\cN)}[\tau_t(\rho)]\;\forall t\in\Z^D\\
    &S(0|\mathcal N)\ge0 \Big\}.
\end{align*}

In general, the MED inequality \eqref{eq:MEDTIhighD} is different from the weak monotonicity inequality \eqref{eq:WMhighD}: \eqref{eq:MEDTIhighD} shows the nonnegativity of a particular conditional entropy $S(0|\cN)$ where $\cN$ has to satisfy the conditions of a Markov shield described above (i.e., compatibility with well-chosen order). Weak monotonicity on the other hand, asserts the nonnegativity of the sum of two conditional entropies \eqref{eq:WMhighD} with the only constraint that $B$ and $C$ are disjoint.

Under some mild conditions however, one can recover the MED inequality \eqref{eq:MEDTIhighD} as a consequence of weak monotonicity. This is the case if the state $\rho$ is reflection symmetric, which can be assumed if the Hamiltonian itself is reflection symmetric (i.e.,  $h=\sum_{i=1}^Dh_{0,e_i}^i$ where $h^i_{0,e_i}$ acts on sites $0,$ $e_i$ and is symmetric under exchange of the two sites). One can check that with the order defined earlier, the reflection of $\cN$ about the origin is disjoint from $\cN$, i.e., $-\cN \cap \cN = \emptyset$. Furthermore, by reflection symmetry we have $S(0|\cN) = S(0|-\cN)$ and so the weak monotonicity inequality $S(0|\cN) + S(0|-\cN) \geq 0$ recovers \eqref{eq:MEDTIhighD}. A corollary of this observation is stated in the next theorem.

\begin{theorem}\label{thm:MEDhigherD}
Under the reflection symmetry and translation-invariance assumptions explained above with a Markov shield $\mathcal N$ of size $l-1$ the MED relaxation is no better than a semidefinite optimization in one $d^{2l-1}\times d^{2l-1}$ dimensional variable, more precisely
\[
\min_{\rho\in\LocTI{0\cup\cN}} \Tr[h\rho]\le\min_{\rho\in\MEDTI{0\cup\cN}} \Tr[h\rho]\le\min_{\rho\in\LocTI{0\cup\cN\cup-\cN}} \Tr[h\rho]\le e_0(h).
\]
\end{theorem}

\begin{proof}
The first inequality is immediate.
We start with the value
\[
\textup{SDP}=\min_{\rho\in\LocTI{0\cup\cN\cup-\cN}} \Tr[h\rho].
\]
Due to the translation-invariance this is the same as
\[
\min_{\rho\in\LocTI{0\cup\cN\cup-\cN}} \Tr[\widetilde h\rho].
\]
with $\widetilde h=\sum_{i=1}^D (h^i_{0,e_i}+h^i_{-e_i,0})/2$ and the variable over the same system.
This problem is intrinsically reflection symmetric, which means that the optimizer $\rho^*$ can be chosen reflection symmetric.
Thereby we also have $S(0|\mathcal N)_{\rho^*}=S(0|-\mathcal N)_{\rho^*}$ and by using weak monotonicity we deduce
\[
S(0|\mathcal N)_{\rho^*}=\frac{S(0|\mathcal N)_{\rho^*}+S(0|-\mathcal N)_{\rho^*}}2\ge0,
\]
which holds because $\mathcal N$ and $-\mathcal N$ are disjoint.
This shows that any optimizer of the SDP on $2l-1$ sites gives a feasible point with same objective value and satisfying the MED inequality. This concludes the proof.
\end{proof}

\section{Numerical experiments}
\label{sec:numerical}
In this section we present numerical experiments for the relaxations with entropy constraints.
The model we consider is the XXZ-Hamiltonian with the nearest-neighbour interaction
\begin{equation}\label{eq:xxzhterm}
h=-\sigma_x\otimes\sigma_x-\sigma_y\otimes\sigma_y-\Delta\sigma_z\otimes\sigma_z
\end{equation}
on an infinite chain.
The model is gapless for $|\Delta|\le1$ and gapped otherwise \cite{franchini2017}.
For $\Delta=0$ this model is also called the XY-Hamiltonian.
We consider the values obtained by the relaxation $\LocTI{l}$ and its strengthening $\WMTI{l}$ based on entropy constraint $S(l|1\ldots l-1) \geq 0$, see Equations \eqref{eq:loctil} and \eqref{eq:WMTIl}.

The resulting convex optimization problems can be solved using different algorithms. One approach is to rely on semidefinite programming solvers, via approximations of the (conditional) entropy function \cite{fawzi2019semidefinite}. Another approach is to use interior-point methods that support the quantum relative entropy cone \cite{fawzi2022optimal,coey2022solving,karimi2020primal}. Finally, another approach is to use custom first-order splitting methods such as \cite{chambolle2011first}, see also \cite{dcbregmanAISTATS}. We have used the latter for the experiments in this section.

We start by comparing the convergence with $l$ for the XY-Hamiltonian, see Figure \ref{fig:MEDvsSDPXY}.
We go up to $l=8$ for the problem with entropy constraints which takes about $900s$ (about $250s$ at $l=7$) and $l=9$ for the standard relaxation $\LocTI{l}$ in about $200s$ (about $20s$ at $l=8$) on a laptop computer. 
The graph shows a significant improvement to the optimal value by using entropy constraints and for a fixed runtime, the entropy-constrained relaxation value outperforms the simpler SDP.

\edit{For comparison, we also add equivalent plots for two gapped models, the XXZ-model at $\Delta=-2$ and the transverse field Ising (TFI) model
\[
h=-\sigma_z\otimes\sigma_z-g\sigma_x
\]
at $g=0.5$ in Figure~\ref{fig:MEDvsSDPTFI}.
While gapped models are typically computationally easier, it is not a priori clear whether the marginal relaxations perform better for these models.
We do, however, observe much faster convergence for the TFI model (while convergence for the XXZ model is comparable to the XY model).
Interestingly, the roughly quadratic speedup in terms of variable sizes persists even in this case.}
\begin{figure}[H]
\centering
\input{MEDvsSDPXYL3-9.tex}
\caption{\label{fig:MEDvsSDPXY} Accuracy of the relaxations with and without entropy constraints for the XY-Hamiltonian on an infinite 1D chain. We verify that the relaxation with entropy constraint gives a value which is not better than the standard SDP relaxation at level $2l-1$. It is surprising to note however that it is better than the value of the SDP for all levels up to $2l-2$.}
\end{figure}
\begin{figure}[H]
\centering
\input{MEDvsSDPTFI0.5L3-9.tex}
\input{MEDvsSDPXXZ-2L3-9.tex}
\caption{\label{fig:MEDvsSDPTFI} \edit{Accuracy of the TFI model at $g=0.5$ (left) and the XXZ model at $\Delta=-2$ (right). Both models are gapped. We observe a similar speedup to the XY-model, even for the much faster convergence rate in case of the TFI.}}
\end{figure}
Returning to the XXZ model we consider the geometry of the set of marginals.
Due to the symmetry of the Hamiltonian, one can show that the two-body marginal of the ground state must be of the form\footnote{More precisely, we use the fact that the Hamiltonian is invariant under the action of the unitaries $U\in\{\sigma_x,\sigma_y,\sigma_z,S\}$ where $S=\exp(i\pi(1-\sigma_z)/4)$, i.e., $U^{\otimes 2}h (U^\dagger)^{\otimes 2} = h$. The set of states \eqref{eq:parametrizedXXZStates} is exactly the set left invariant by this action.}
\begin{equation}\label{eq:parametrizedXXZStates}
\rho(x,z)=\frac14(\id\otimes\id+x(\sigma_x\otimes\sigma_x+\sigma_y\otimes\sigma_y)+z\sigma_z\otimes\sigma_z).
\end{equation}

In Figure~\ref{fig:MEDvsSDPStateSet} we compare the convex sets $\TI$ with $\LocTI{l}$ and $\WMTI{l}$ on the two-dimensional slice \eqref{eq:parametrizedXXZStates}, as done in \cite{Verstraete2006}, i.e.,
we plot the set of valid states
\[
\{(x,z)\in\mathbb{R}^2: \rho(x,z)\in\TI\},
\]
which we obtain from the analytic solution for infinite systems \cite{Yang1966} with its relaxations
\[
\{(x,z)\in\mathbb{R}^2: \rho(x,z)\in\hat\cC\}.
\]
\edit{We depict the approximation errors in Figure~\ref{fig:EnergyErrorsAni}.}

\begin{figure}[H]
\centering
\input{stateSetAni.tex}
\caption{\label{fig:MEDvsSDPStateSet} Two-dimensional slice~\eqref{eq:parametrizedXXZStates} of sets of feasible 2-body marginals.
The blue lines are for the MED, the red lines for the unconstrained marginal relaxation, and the green line is for the exact analytic solution.
As the size of the marginals increases the outer approximations become better.
\edit{Note that we find numerically that $\LocTI{3}=\LocTI{4}$. We obtain almost the same  optimal speedups as in Figure~\ref{fig:MEDvsSDPXY}.}}
\end{figure}
\begin{figure}[H]
\centering
\input{energyErrorsAni.tex}
\caption{\label{fig:EnergyErrorsAni}\edit{The energy errors corresponding to the nontrivial part of the boundary of the state sets in Figure~\ref{fig:MEDvsSDPStateSet}. 
}}
\end{figure}

\section*{Acknowledgements}

We would like to thank Matt Hastings for bringing the paper \cite{Poulin2011} to our attention as well as an anonymous referee for multiple comments that improved the paper.
HF acknowledges funding from UK Research and Innovation (UKRI) under the UK
government’s Horizon Europe funding guarantee EP/X032051/1. OF acknowledges funding by the European Research Council (ERC Grant AlgoQIP, Agreement No. 851716). SOS acknowledges support from the UK Engineering and Physical Sciences Research Council (EPSRC) under grant number EP/W524141/1.

\section*{Author declarations}
The authors have no conflicts to disclose.
The data that support the findings of this study are available from the corresponding author upon reasonable request.

\appendix
\counterwithin{equation}{section}
\renewcommand{\theequation}{A\arabic{equation}}
\section{Edge covering for connected graphs}\label{sec:adjEdge}
We give the proof of Proposition \ref{prop:adjEdge}.
\begin{proof}
The proof is by induction over the number of unmatched edges and the minimum distance between any pair of unmatched edges.
We start from a set $P=\{\{e_1,e_1'\},\ldots\{e_m,e_m'\}\}$ of pairs of edges, where $e_i$ is adjacent to $e_i'$ for all $i$ and the $e_1,e_1',e_2,\ldots e_m'$ are all distinct.
We denote by $\widetilde P$ the set of all edges contained in elements of $P$.
If $E\setminus \widetilde P$ contains less than 2 edges, we are done.
Otherwise, consider the pair of edges $f, g\in E\setminus\widetilde P$ of minimal distance.
If the distance is 0, add $\{f, g\}$ to $P$ and repeat.
If the distance is greater zero, there is a shortest path $f_1,\ldots, f_l$ from $f$ to $g$, with $f_i\in\widetilde P$ (otherwise if $f_i\notin\widetilde P$ then the pair $f_i$, $g$ has smaller distance).
For some $i$, $f_1=e_i$ (without loss of generality we can exchange $e_i$ and $e_i'$).
We remove $\{e_i,e_i'\}$ from $P$ and add $\{e_i',f\}$ if they are adjacent and $\{e_i,f\}$ otherwise. The number of unmatched edges is unchanged but the unmatched edges $e_i'$ or $e_i$ and $g$ are now connected by the path $f_2,\ldots,f_l$, so the minimum distance of unmatched edges has decreased by at least one.
We repeat this procedure until the distance goes to zero, after which we can add a pair of edges and repeat until only one edge is unmatched.
\end{proof}

\section{Definition of the translation-invariant ground energy}\label{sec:existence}
\renewcommand{\theequation}{B\arabic{equation}}
We prove the existence of the limit ground energy by standard arguments \cite{Ohya1993,Eisert2023}.

\begin{theorem}\sloppy
    The following limits exist and are equal
    \begin{equation}
        \label{eq:thmB2}
    \lim_{m\to\infty} \frac{1}{m^D}\lambda_{min}\left(H_{[1,m]^D}\right)=\lim_{m\to\infty}\min_{\rho_m\;TI}\Tr[\rho_m h]=\min_{\rho\in\TI}\Tr[\rho h],
    \end{equation}
    where the minimization in the middle expression goes over all translation-invariant states in $\dens(\otimes_{v\in[1,m]^D}\CC^d)$.
\end{theorem}
\begin{proof}
We first note that the existence of the limit in the middle expression is trivial since it is an increasing sequence.
Indeed, a translation-invariant state on $[1,m]^D$ defines a translation-invariant state on $[1,m-1]^D$.

We prove the first equality. 
We rewrite the left-hand side as a minimzation:
\[
    \frac{1}{m^D}\lambda_{min}\left(H_{[1,m]^D}\right)=\frac{1}{m^D}\min_{\rho'_m}\Tr[\rho'_m H_{[1,m]^D}]
\]
over all states $\rho'_m\in \dens(\otimes_{v\in[1,m]^D}\CC^d)$.
Clearly, the minimizer $\rho'_m$ can be chosen periodic 
\[
    \frac{1}{m^D}\min_{\rho'_m}\Tr[\rho'_m H_{[1,m]^D}]=\frac{1}{m^D}\min_{\rho'_m \textrm{ periodic}}\Tr[\rho'_m H_{[1,m]^D}]=\min_{\rho'_m \textrm{ periodic}}\Tr[\rho'_m h]
\]
meaning that $\rho'_m$ is invariant under a permutation of sites mapping $v\to v+t$ modulo $m$ $\forall t\in[1,m]^D$, which is a stronger condition than translation invariance. Therefore,
\[
    \frac{1}{m^D}\lambda_{min}\left(H_{[1,m]^D}\right)\ge\min_{\rho_m\;TI}\Tr[\rho_m h].
\]
For the reverse direction we note that
\[
\min_{\rho_m\;TI}\Tr[\rho_m h]=\frac{1}{(m-1)^D}\min_{\rho_m\;TI}\Tr[\rho_m H_{[1,m-1]^D}]
\]
due to translation-invariance and
\begin{align*}
\left\|\frac{1}{m^D}H_{[1,m]^D}-\frac{1}{(m-1)^D}H_{[1,m-1]^D}\right\|&\le\left\|\frac{1}{m^D}H_{[1,m]^D}-\frac{1}{m^D}H_{[1,m-1]^D}\right\|\\
&\quad+\left\|\frac{1}{m^D}H_{[1,m-1]^D}-\frac{1}{(m-1)^D}H_{[1,m-1]^D}\right\|\\
&\le\|h\| \left(1-\left(\frac{m-1}{m}\right)^D\right)+\left\|H_{[1,m-1]^D}\right\|\left(\frac{1}{(m-1)^D}-\frac{1}{m^D}\right)\\
&=2\|h\|\left(1-\left(\frac{m-1}{m}\right)^D\right).
\end{align*}
We conclude the following bound
\begin{align*}
\min_{\rho_m\;TI}\Tr\left[\rho_m\frac{H_{[1,m-1]^D}}{(m-1)^D}\right]&\ge\min_{\rho_m\;TI} \Tr\left[\rho_m\frac{H_{[1,m]^D}}{m^D}\right]-\left\|\frac{H_{[1,m]^D}}{m^D}-\frac{H_{[1,m-1]^D}}{(m-1)^D}\right\|\\
&\ge\frac{1}{m^D} \lambda_{\min}(H_{[1,m]^D})-2\|h\|\left(1-\left(\frac{m-1}{m}\right)^D\right),
\end{align*}
where the first inequality follows from lower bounding the objective, and in the second inequality we drop the translation-invariance constraint. Since the last term decays with $m$, the existence of limits and the first equality follow.

To prove the second equality in \eqref{eq:thmB2}, we note that one inequality is again trivial
\[
\min_{\rho\in\TI}\Tr[\rho h]\ge\lim_{m\to\infty}\min_{\rho_m\;TI}\Tr[\rho_m h]
\]
as the minimizer in $\TI$ defines finite translation-invariant states of any size and thereby feasible points for every element of the sequence.

To prove the reverse inequality we need to construct an infinite translation-invariant state from a finite one $\rho_m$ that achieves the same objective value up to some error that decays with $m$.
Let us consider the sequence of states
\[
\rho_{m,n}=\frac{1}{m^D}\sum_{i\in[1,m]^D}\tau_i\left(\bigotimes_{j=1}^n \rho_m\right)
\]
defined on $\dens\left(\otimes_{i\in[1,nm]^D}\CC^d\right)$. $\tau_i$ is the operator translating the system by $i$ sites (with periodic boundary conditions).
Since these states are all periodic and hence translation-invariant and have identical marginals on sites $[1,m]$ for any $n\ge2$, it defines an element of $\TI$.
Furthermore, for the objective value we have, using the notation $\vec 1=(1,\ldots,1)$,
\begin{align*}
\left|\Tr[\rho_{m,n} h]-\Tr[\rho_m h]\right|&=\left|\frac{1}{m^D}\sum_{i\in[1,m]^D}\Tr\left[\left(\tau_i\left(\bigotimes_{j=1}^n \rho_m\right)-\rho_m\right)h\right]\right|\\
&=\left|\frac{1}{m^D}\sum_{i=1}^D\Tr\left[\left(\tau_{\vec1+e_i}\left(\bigotimes_{j=1}^n \rho_m\right)-\rho_m \right)h\right]\right|\\
&\le \frac{2D}{m^D}\|h\|
\end{align*}
where the only nonzero term in the sum appear, where the marginal on sites $\vec1$ and $\vec1+e_i$ come from two different copies of $\rho_m$.
\end{proof}
The above proof also implies a convergence rate of $O(1/m)$ for both limits.

\bibliographystyle{ieeetr}
\bibliography{refs}
\end{document}

%% file: starSet.tex
%
%
\begin{tikzpicture}

\begin{axis}[%
width=4.22in,
height=3.33in,
scale only axis,
xmin=0,
xmax=1,
xlabel style={font=\color{white!15!black}},
xlabel={$\tr[(\rho_{12}+\rho_{23})/2 \ket{\psi^{-}} \bra{\psi^{-}}]$},
ymin=0,
ymax=1,
ylabel style={font=\color{white!15!black}},
ylabel={$\tr[(\rho_{12}+\rho_{23})/2 \ket{\phi^{+}} \bra{\phi^{+}}]$},
axis background/.style={fill=white},
axis x line*=bottom,
axis y line*=left,
legend pos=north east
]
\addplot [color=blue,ultra thick]
  table[row sep=crcr]{%
0.0991069377654501	0.749315757154563\\
0.111909007604277	0.747698479629001\\
0.127144600658905	0.74479213481698\\
0.145469978262456	0.740086977561128\\
0.167722531390783	0.732856684758266\\
0.194936279594035	0.72208199962023\\
0.22830343552312	0.706380607874364\\
0.269009403351175	0.68400231268888\\
0.31783704619191	0.653015348288555\\
0.374482269624952	0.611860184456084\\
0.436781181806248	0.560322042143896\\
0.500492825947167	0.500492825947172\\
0.560322042143902	0.436781181806241\\
0.611860184456079	0.374482269624959\\
0.653015348288555	0.31783704619191\\
0.684002312688881	0.269009403351173\\
0.706380607874363	0.228303435523118\\
0.722081999620229	0.194936279594035\\
0.732856684758267	0.167722531390778\\
0.740086977561128	0.145469978262458\\
0.744792134816979	0.127144600658909\\
0.747698479629002	0.111909007604272\\
0.749315757154564	0.0991069377654483\\
0.749999991824368	0.0882313391225154\\
0.749999991824368	0.0788919776441155\\
0.749490127402386	0.070787914724626\\
0.7485929323747	0.0636858847873167\\
0.747394559310464	0.0574037931859319\\
0.745955335683692	0.0517983872489154\\
0.744316983092957	0.0467560564523239\\
0.74250757461734	0.0421860123144704\\
0.740544932425864	0.0380151853523838\\
0.73843895396709	0.0341844247569529\\
0.736193166931003	0.030645630640894\\
0.733805751960797	0.0273596358392405\\
0.731270123799752	0.0242945879146861\\
0.72857520906742	0.0214247952058555\\
0.725705421504879	0.0187298853062101\\
0.722640373148413	0.0161942567878572\\
0.719354388644854	0.0138068492996553\\
0.715815593378391	0.0115610615334996\\
0.711984832606565	0.00945508297775088\\
0.707814004291128	0.0074924401494374\\
0.703243700901181	0.00568292902865442\\
0.698201614148828	0.00404465573269996\\
0.692596210231749	0.0026054326245597\\
0.686314123983745	0.00140706058153611\\
0.679212092886106	0.000509865407266316\\
0.671108033563963	1.21160990602408e-09\\
-8.40221161014772e-09	1.21161012186889e-09\\
-2.51087605917442e-09	8.40958705999804e-10\\
-2.23525996976864e-09	8.06140296735426e-10\\
-1.51124483191324e-09	6.68027022628826e-10\\
-7.61507773937612e-10	4.75527264390195e-10\\
-3.92572834563462e-10	3.55653035959041e-10\\
-2.3945895998368e-10	2.95030964477363e-10\\
-1.46573494889877e-10	2.51322382360433e-10\\
-8.41132203357877e-11	2.16984555847084e-10\\
-2.81916277568827e-11	1.8149563404711e-10\\
3.24743236523248e-12	1.58653819827875e-10\\
5.55372028382537e-11	1.15395959654205e-10\\
8.64084387717083e-11	8.64059394808522e-11\\
1.15394038292873e-10	5.5539411063842e-11\\
1.58654748401907e-10	3.24619561075405e-12\\
1.81495426799923e-10	-2.81913011875721e-11\\
2.16984541043714e-10	-8.41131970094094e-11\\
2.51322567146793e-10	-1.46573834615871e-10\\
2.95031220445655e-10	-2.39459450978977e-10\\
3.55652937293659e-10	-3.92572429856345e-10\\
4.75527102301823e-10	-7.61507174034601e-10\\
6.68026887900817e-10	-1.51124433857181e-09\\
8.06140399670899e-10	-2.23526072230282e-09\\
8.40958749266911e-10	-2.51087633938484e-09\\
1.21161006310088e-09	-8.40221026855646e-09\\
1.21160993982033e-09	0.671108033563964\\
0.000509865407266344	0.679212092886106\\
0.00140706058153628	0.686314123983747\\
0.00260543262455948	0.692596210231748\\
0.00404465573269996	0.698201614148828\\
0.00568292902865397	0.703243700901179\\
0.00749244014943873	0.70781400429113\\
0.00945508297775044	0.711984832606564\\
0.0115610615334987	0.715815593378389\\
0.0138068463934848	0.719354384065462\\
0.0161942602450642	0.722640377327459\\
0.0187298853062074	0.725705421504876\\
0.0214247952058564	0.728575209067422\\
0.0242945879146808	0.731270123799748\\
0.0273596358392449	0.733805751960801\\
0.0306456306408904	0.736193166931001\\
0.0341844247569565	0.738438953967091\\
0.0380151846053298	0.74054493201517\\
0.0421860131071075	0.742507574931167\\
0.046756057076669	0.744316983340154\\
0.051798386598513	0.745955335516698\\
0.0574037931859284	0.747394559310464\\
0.0636858847873167	0.7485929323747\\
0.0707879147246224	0.749490127402386\\
0.0788919776441173	0.749999991824368\\
0.08823133912251	0.749999991824368\\
0.0991069377654501	0.749315757154563\\
};\addlegendentry{$\cC$}
\addplot [color=red,ultra thick]
  table[row sep=crcr]{%
0.0742550014279857	0.810229425476812\\
0.0829913026167119	0.809125773977722\\
0.0931069834340281	0.807196104659432\\
0.104994404473642	0.804143933699143\\
0.119203175911171	0.799527223999818\\
0.136525984902553	0.792668638729077\\
0.158145378716616	0.782495324218718\\
0.18590507865483	0.767234300034209\\
0.222851770214668	0.743787216589925\\
0.274488116709335	0.706271214870728\\
0.352733056182354	0.641541371530384\\
0.501675297335939	0.501675297330627\\
0.641541371527247	0.35273305618604\\
0.706271214871147	0.274488116708721\\
0.743787216591429	0.222851770212561\\
0.767234300023179	0.185905078672475\\
0.782495324222354	0.158145378707582\\
0.792668638731131	0.136525984896879\\
0.799527224076609	0.119203175716734\\
0.804143933648574	0.104994404671178\\
0.807196104659004	0.0931069834362841\\
0.809125774182417	0.0829913015436752\\
0.810229425407608	0.0742550025230759\\
0.810710355448321	0.066610838424209\\
0.810710355448321	0.059844789328098\\
0.810329676472085	0.0537940702692072\\
0.809639777660063	0.0483329588435932\\
0.808691780739363	0.0433633849509558\\
0.807522151898429	0.0388079813798079\\
0.80615664716156	0.0346053899312118\\
0.8046130686536	0.0307067556500762\\
0.802903369892405	0.0270734608390271\\
0.801035330570047	0.0236755099259991\\
0.799013906052618	0.020490254568573\\
0.796842788629903	0.0175019678007287\\
0.794526332529097	0.0147018533860157\\
0.792071576142811	0.0120878033402745\\
0.789492045018418	0.00966546237876731\\
0.78681272931414	0.00744893966221483\\
0.784077624864254	0.00546176996083547\\
0.781360031507556	0.00373713277379695\\
0.778774758388259	0.00231586684927398\\
0.776484942697088	0.00123836137445044\\
0.774686877502663	0.000526457002333336\\
0.773538984971256	0.000153484109720159\\
0.773030611894814	2.29560889040137e-05\\
0.77291513376665	9.27457823762801e-07\\
0.772907805626568	1.69844198427693e-09\\
0.772907792057641	8.44757497462467e-10\\
7.63867349607516e-08	8.44757746048548e-10\\
-5.8393833631073e-08	9.3244323712609e-09\\
1.97825195577512e-10	1.92258452956372e-09\\
2.13345049841552e-09	1.55334425537496e-09\\
1.98482365479103e-09	1.59150514279693e-09\\
-1.37682673521243e-09	2.68377156635809e-09\\
4.36597737138061e-09	4.10034571519324e-10\\
1.17121803913783e-08	-3.04682617222516e-09\\
3.46177252476129e-09	1.48887393493537e-09\\
2.93479443271933e-09	1.82330440152819e-09\\
2.46232507196051e-09	2.16657348529899e-09\\
3.04227505375737e-09	1.68679713529119e-09\\
2.34362928160185e-09	2.34286918477035e-09\\
1.71517941356962e-09	3.01210032129107e-09\\
2.16750570994905e-09	2.4653317543385e-09\\
1.80199605456932e-09	2.9684126357608e-09\\
1.48843476877866e-09	3.46250616081446e-09\\
-3.04445415529872e-09	1.17078005043864e-08\\
4.10485916493944e-10	4.36567912005142e-09\\
2.68411776058264e-09	-1.37685940606304e-09\\
1.59175466180707e-09	1.98508851965629e-09\\
1.55379962895098e-09	2.13291361268116e-09\\
1.92393808411608e-09	1.92579880276114e-10\\
9.32381202335136e-09	-5.83834538994838e-08\\
8.447588658069e-10	7.63872367466246e-08\\
8.44758723826056e-10	0.77290779205792\\
1.69840896901974e-09	0.772907805626302\\
9.27457850352642e-07	0.772915133766856\\
2.29560870567136e-05	0.773030611885196\\
0.000153484109623792	0.773538984968457\\
0.000526457006680525	0.77468687751354\\
0.00123836137485633	0.77648494269801\\
0.00231586685919671	0.778774758409405\\
0.00373713278985566	0.781360031539863\\
0.00546176992433312	0.784077624813738\\
0.00744881221829541	0.786812553902361\\
0.00966561080454387	0.789492203075818\\
0.0120878033353611	0.792071576137579\\
0.0147018534097727	0.794526332550789\\
0.0175019677937316	0.796842788626153\\
0.0204902546020982	0.799013906078309\\
0.0236755098729766	0.801035330540811\\
0.0270734608397287	0.802903369892706\\
0.0307067556502663	0.804613068653659\\
0.0346053899333878	0.806156647162406\\
0.0388079813774596	0.807522151897805\\
0.0433633849515047	0.808691780739483\\
0.0483329588280359	0.809639777657111\\
0.0537940702976965	0.810329676474698\\
0.0598447893165552	0.810710355448414\\
0.0666108384275965	0.810710355448414\\
0.0742550014279857	0.810229425476812\\
};\addlegendentry{$\WM{2}$}
\addplot [color=black,ultra thick]
  table[row sep=crcr]{%
0	1\\
1	0\\
};\addlegendentry{$\Loc{2}$}
\addplot [color=black,ultra thick,forget plot]
  table[row sep=crcr]{%
0	1\\
0	0\\
};
\addplot [color=black,ultra thick,forget plot]
  table[row sep=crcr]{%
0	0\\
1	0\\
};
\addplot [dash pattern=on 3pt off 2pt, color=orange, ultra thick]
   table[row sep=crcr]{%
 0.75	0\\
 0.75	0.25\\
 };\addlegendentry{linear (Eq.~\eqref{eq:starBound})}
\end{axis}

\end{tikzpicture}%

%% file: markovShield.tex
\usetikzlibrary{decorations.markings}
\tikzset{->-/.style={decoration={
  markings,
  mark=at position .5 with {\arrow{>}}},postaction={decorate}}}

\begin{tikzpicture}
\fill[orange!80] (0.5,2.5) rectangle(3.5,3.5);
\fill[orange!80] (1.5,3.5) rectangle(2.5,4.5);
\fill[orange!80] (2.5,1.5) rectangle(4.5,2.5);
\fill[blue!50] (1.5,1.5) rectangle(2.5,2.5);
\draw[black,thick] (4,4)--(3,4);
\draw[->-,black,thick] (3,4)--(0,4);
\draw[->-,black] (0,4)--(4,3);
\draw[black,thick] (4,3)--(3,3);
\draw[->-,black,thick] (3,3)--(0,3);
\draw[->-,black] (0,3)--(4,2);
\draw[black,thick] (4,2)--(3,2);
\draw[->-,black,thick] (3,2)--(2,2);
\filldraw[black]{(0,0)} circle (4pt);
\filldraw[black]{(0,1)} circle (4pt);
\filldraw[black]{(0,2)} circle (4pt);
\filldraw[black]{(0,3)} circle (4pt);
\filldraw[black]{(0,4)} circle (4pt);
\filldraw[black]{(1,0)} circle (4pt);
\filldraw[black]{(1,1)} circle (4pt);
\filldraw[black]{(1,2)} circle (4pt);
\filldraw[black]{(1,3)} circle (4pt);
\filldraw[black]{(1,4)} circle (4pt);
\filldraw[black]{(2,0)} circle (4pt);
\filldraw[black]{(2,1)} circle (4pt);
\filldraw[black]{(2,2)} circle (4pt)node[anchor=north east]{ $0$};
\filldraw[black]{(2,3)} circle (4pt)node[anchor=north east]{ $e_2$};
\filldraw[black]{(2,4)} circle (4pt);
\filldraw[black]{(3,0)} circle (4pt);
\filldraw[black]{(3,1)} circle (4pt);
\filldraw[black]{(3,2)} circle (4pt)node[anchor=north east]{ $e_1$};
\filldraw[black]{(3,3)} circle (4pt);
\draw[black]{(2.8,3.2)} circle (0pt) node[anchor=east]{ $\mathcal N$};
\filldraw[black]{(3,4)} circle (4pt);
\filldraw[black]{(4,0)} circle (4pt);
\filldraw[black]{(4,1)} circle (4pt);
\filldraw[black]{(4,2)} circle (4pt);
\filldraw[black]{(4,3)} circle (4pt);
\filldraw[black]{(4,4)} circle (4pt);
\end{tikzpicture}

%% file: MEDvsSDPXYL3-9.tex
%
%
\begin{tikzpicture}

\begin{axis}[%
width=4.22in,
height=3.33in,
scale only axis,
xmin=3,
xmax=9,
xlabel style={font=\color{white!15!black}},
xlabel={$l$},
xmode=log,
xtick={3,4,5,6,7,8,9},
xticklabels={3,4,5,6,7,8,9},
ymin=0.008,
ymax=0.15,
ylabel style={font=\color{white!15!black}},
ylabel={$e_0 - \hat{e}$},
axis background/.style={fill=white},
axis x line*=bottom,
axis y line*=left,
ymode=log,
]
\addplot [color=red, very thick]
  table[row sep=crcr]{%
3	0.141016130535025\\
4	0.14107728113447\\
5	0.0600937822357372\\
6	0.0601905074573552\\
7	0.0333290622071962\\
8	0.0333228948179121\\
9	0.0211472334540042\\
}; \addlegendentry{$\LocTI{l}$ (w/o entropy constraints)}
\addplot [color=blue, very thick]
  table[row sep=crcr]{%
3	0.0884756675434113\\
4	0.0438847550199803\\
5	0.0263988628195859\\
6	0.0177077538405186\\
7	0.0126865223570045\\
8	0.00958114911505192\\
}; \addlegendentry{$\WMTI{l}$ (with entropy constraints)}
\end{axis}

\end{tikzpicture}%

%% file: MEDvsSDPTFI0.5L3-9.tex
%
%
\begin{tikzpicture}

\begin{axis}[%
width=2.4in,
height=2in,
scale only axis,
xmode=log,
xmin=3,
xmax=9,
xtick={3,4,5,6,7,8,9},
xticklabels={3,4,5,6,7,8,9},
xminorticks=true,
xlabel style={font=\color{white!15!black}},
xlabel={$l$},
ymode=log,
ymin=1e-05,
ymax=0.0138058577061693,
yminorticks=true,
ylabel style={font=\color{white!15!black}},
ylabel={$e_0 - \hat{e}$},
axis background/.style={fill=white},
axis x line*=bottom,
axis y line*=left
]
\addplot [color=red]
  table[row sep=crcr]{%
3	0.0138058577061693\\
4	0.0043451917086379\\
5	0.00152420351121618\\
6	0.000571351451551294\\
7	0.000224111718668762\\
8	9.08694267665044e-05\\
9	3.7788464355204e-05\\
};\addlegendentry{$\LocTI{l}$}
\addplot [color=blue]
  table[row sep=crcr]{%
3	0.00156785718148789\\
4	0.000242815563757715\\
5	4.64628311558357e-05\\
6	1.61678862358183e-05\\
7	1.17924702998096e-05\\
8	-8.46576264201815e-06\\
};\addlegendentry{$\WMTI{l}$}
\end{axis}

\end{tikzpicture}%

%% file: MEDvsSDPXXZ-2L3-9.tex
%
%
\begin{tikzpicture}

\begin{axis}[%
width=2.4in,
height=2in,
scale only axis,
xmode=log,
xmin=3,
xmax=9,
xtick={3,4,5,6,7,8,9},
xticklabels={3,4,5,6,7,8,9},
xminorticks=true,
xlabel style={font=\color{white!15!black}},
xlabel={$l$},
ymode=log,
ymin=0.01,
ymax=0.263162630163822,
yminorticks=true,
ylabel style={font=\color{white!15!black}},
ylabel={$e_0 - \hat{e}$},
axis background/.style={fill=white},
axis x line*=bottom,
axis y line*=left
]
\addplot [color=red]
  table[row sep=crcr]{%
3	0.263162625645832\\
4	0.263162630163822\\
5	0.10162682414873\\
6	0.101626820543035\\
7	0.0508261683808451\\
8	0.050826168074332\\
9	0.0289452785822482\\
};\addlegendentry{$\LocTI{l}$}
\addplot [color=blue ]
  table[row sep=crcr]{%
3	0.157550373509086\\
4	0.0705342951551127\\
5	0.0385078187077763\\
6	0.0235037483819722\\
7	0.0150632505270978\\
8	0.0106350568378826\\
};\addlegendentry{$\WMTI{l}$}
\end{axis}

\end{tikzpicture}%

%% file: stateSetAni.tex
%
%
\begin{tikzpicture}

\begin{axis}[%
width=5in,
height=4in,
scale only axis,
xmin=-1,
xmax=1,
xlabel style={font=\color{white!15!black}},
xlabel={$z$},
ymin=-1,
ymax=1,
ylabel style={font=\color{white!15!black}},
ylabel={$x$},
axis background/.style={fill=white},
axis x line*=bottom,
axis y line*=left
]
\addplot [color=black, forget plot]
  table[row sep=crcr]{%
-1	1\\
1	0\\
};
\addplot [color=black, forget plot]
  table[row sep=crcr]{%
-1	1\\
-1	-1\\
};
\addplot [color=black, forget plot]
  table[row sep=crcr]{%
-1	-1\\
1	0\\
};
\addplot [dash pattern=on 3pt off 2pt, color=red]
  table[row sep=crcr]{%
-1	-0\\
-0.993427369672116	-0.114271992677065\\
-0.992774535283655	-0.119821084978984\\
-0.991969763218203	-0.1262592615026\\
-0.990610299232884	-0.136455241392489\\
-0.989415389530125	-0.144819609311804\\
-0.987598801285493	-0.156627432901908\\
-0.985532609231017	-0.169024585228767\\
-0.982814729456985	-0.183972923985943\\
-0.979013802058498	-0.202977560978377\\
-0.976705892046151	-0.213363156033937\\
-0.971060698953442	-0.238202005641856\\
-0.978901338582826	-0.204487255235507\\
-0.975549737168645	-0.218563981175065\\
-0.968620728803419	-0.246972915472494\\
-0.973402801028581	-0.227844626571844\\
-0.967695971314177	-0.250101262458023\\
-0.968556143848661	-0.246832606826981\\
-0.966033012085673	-0.256168194350037\\
-0.96446307545368	-0.261819966225212\\
-0.962694124658911	-0.268011294006905\\
-0.960556641651443	-0.275278736232293\\
-0.95850190861868	-0.282059355240414\\
-0.953124388309083	-0.299267420231125\\
-0.95657371786181	-0.288574498617671\\
-0.947766542753347	-0.314996023943059\\
-0.953142184068544	-0.299406664128987\\
-0.941510018224773	-0.331976728491547\\
-0.942764685086618	-0.328589127964565\\
-0.936368582673829	-0.345218994237816\\
-0.933428097161601	-0.352570208018387\\
-0.928161228474211	-0.365210692868123\\
-0.923390795498727	-0.376182688711734\\
-0.917639058832131	-0.388836509378247\\
-0.911546805525666	-0.401630241321823\\
-0.904700325984579	-0.415323200403997\\
-0.897183582591453	-0.429605012850936\\
-0.886832384866919	-0.448237168755098\\
-0.882817633924108	-0.455062245357876\\
-0.868218825710948	-0.478420338498932\\
-0.858081330064479	-0.493626581968636\\
-0.845171057564723	-0.511700963468294\\
-0.831106824418149	-0.52998446655884\\
-0.815407495008759	-0.548823661850108\\
-0.797629044568785	-0.568379957334079\\
-0.778435588698587	-0.587573413204278\\
-0.757095695692424	-0.606779316909825\\
-0.734744402521506	-0.624660351446559\\
-0.710099583620207	-0.641911724677468\\
-0.679714593850684	-0.660142718539182\\
-0.651547612841642	-0.674226209043703\\
-0.620060558784657	-0.686821030666497\\
-0.586909448847016	-0.696766363647789\\
-0.552617769605867	-0.703624699496019\\
-0.51758638821512	-0.707127837635094\\
-0.482415285257867	-0.707127837635094\\
-0.447280933513734	-0.703614402460681\\
-0.412794255041547	-0.696717066766243\\
-0.37969392156204	-0.686786966722391\\
-0.34910399106949	-0.674550994525371\\
0	-0.5\\
};\addlegendentry{$\LocTI{3}$}
\addplot [dash pattern=on 3pt off 2pt, color=red, forget plot]
  table[row sep=crcr]{%
-1	0\\
-0.993427369672116	0.114271992677065\\
-0.992774535283655	0.119821084978984\\
-0.991969763218203	0.1262592615026\\
-0.990610299232884	0.136455241392489\\
-0.989415389530125	0.144819609311804\\
-0.987598801285493	0.156627432901908\\
-0.985532609231017	0.169024585228767\\
-0.982814729456985	0.183972923985943\\
-0.979013802058498	0.202977560978377\\
-0.976705892046151	0.213363156033937\\
-0.971060698953442	0.238202005641856\\
-0.978901338582826	0.204487255235507\\
-0.975549737168645	0.218563981175065\\
-0.968620728803419	0.246972915472494\\
-0.973402801028581	0.227844626571844\\
-0.967695971314177	0.250101262458023\\
-0.968556143848661	0.246832606826981\\
-0.966033012085673	0.256168194350037\\
-0.96446307545368	0.261819966225212\\
-0.962694124658911	0.268011294006905\\
-0.960556641651443	0.275278736232293\\
-0.95850190861868	0.282059355240414\\
-0.953124388309083	0.299267420231125\\
-0.95657371786181	0.288574498617671\\
-0.947766542753347	0.314996023943059\\
-0.953142184068544	0.299406664128987\\
-0.941510018224773	0.331976728491547\\
-0.942764685086618	0.328589127964565\\
-0.936368582673829	0.345218994237816\\
-0.933428097161601	0.352570208018387\\
-0.928161228474211	0.365210692868123\\
-0.923390795498727	0.376182688711734\\
-0.917639058832131	0.388836509378247\\
-0.911546805525666	0.401630241321823\\
-0.904700325984579	0.415323200403997\\
-0.897183582591453	0.429605012850936\\
-0.886832384866919	0.448237168755098\\
-0.882817633924108	0.455062245357876\\
-0.868218825710948	0.478420338498932\\
-0.858081330064479	0.493626581968636\\
-0.845171057564723	0.511700963468294\\
-0.831106824418149	0.52998446655884\\
-0.815407495008759	0.548823661850108\\
-0.797629044568785	0.568379957334079\\
-0.778435588698587	0.587573413204278\\
-0.757095695692424	0.606779316909825\\
-0.734744402521506	0.624660351446559\\
-0.710099583620207	0.641911724677468\\
-0.679714593850684	0.660142718539182\\
-0.651547612841642	0.674226209043703\\
-0.620060558784657	0.686821030666497\\
-0.586909448847016	0.696766363647789\\
-0.552617769605867	0.703624699496019\\
-0.51758638821512	0.707127837635094\\
-0.482415285257867	0.707127837635094\\
-0.447280933513734	0.703614402460681\\
-0.412794255041547	0.696717066766243\\
-0.37969392156204	0.686786966722391\\
-0.34910399106949	0.674550994525371\\
0	0.5\\
};
\addplot [dash pattern=on 3pt off 2pt,color=blue]
  table[row sep=crcr]{%
-1	-0\\
-0.99759812984569	-0.0634584447120199\\
-0.99314829192307	-0.094607310170356\\
-0.991128884694668	-0.107733457154974\\
-0.989370729231402	-0.11828238993457\\
-0.987156160081357	-0.130462520259816\\
-0.984197650866276	-0.145255066335218\\
-0.981976569024304	-0.155249934624092\\
-0.981115849644926	-0.159037099893356\\
-0.980194038548591	-0.163000887607598\\
-0.979205404824194	-0.167153149250064\\
-0.978143553328547	-0.171506740382219\\
-0.977001333895545	-0.176075618114224\\
-0.975770736456098	-0.18087494812807\\
-0.974442769681899	-0.185921221870024\\
-0.973007320340543	-0.191232384433041\\
-0.971452990063687	-0.196827973429723\\
-0.969766905658068	-0.202729268849389\\
-0.967934498439619	-0.208959453392117\\
-0.965939247337516	-0.215543782029056\\
-0.963762379697734	-0.22250975847636\\
-0.961382522847226	-0.229887314712933\\
-0.958775298587157	-0.237708987493141\\
-0.955912851950104	-0.246010082740595\\
-0.952763304909315	-0.254828814454803\\
-0.949290125479817	-0.264206398914448\\
-0.945451403138446	-0.274187077002014\\
-0.941199023235727	-0.28481802675881\\
-0.936477736865097	-0.296149114048322\\
-0.931224129655156	-0.308232410631187\\
-0.925365504824573	-0.32112138525847\\
-0.918818714900306	-0.33486964409943\\
-0.911489005798657	-0.349529062302728\\
-0.903268980286513	-0.365147110775802\\
-0.894037849306398	-0.381763146540008\\
-0.883661222737217	-0.399403411707617\\
-0.871991796429096	-0.41807449380061\\
-0.858871413280568	-0.437755068523402\\
-0.844135092308348	-0.45838591788451\\
-0.827617687479945	-0.479858544161434\\
-0.809163781139607	-0.502003231769839\\
-0.788641123171331	-0.524578155534943\\
-0.765957266520489	-0.547262012185785\\
-0.741077929952977	-0.569653415096546\\
-0.714044088461155	-0.591280488290004\\
-0.684983161622723	-0.611623137076906\\
-0.654108557335916	-0.63014789964899\\
-0.621701971532707	-0.646351192550595\\
-0.588074617118414	-0.659802134316312\\
-0.553506366597155	-0.67017260947269\\
-0.518163921687246	-0.677241098454671\\
-0.481998844841091	-0.680857606139287\\
-0.444626957256001	-0.680857606139287\\
-0.405227964998137	-0.676917706913501\\
-0.36282270849924	-0.668436655613721\\
-0.318745292882476	-0.655213430928692\\
0	-0.5\\
};\addlegendentry{$\WMTI{3}$}
\addplot [dash pattern=on 3pt off 2pt, color=blue, forget plot]
  table[row sep=crcr]{%
-1	0\\
-0.99759812984569	0.0634584447120199\\
-0.99314829192307	0.094607310170356\\
-0.991128884694668	0.107733457154974\\
-0.989370729231402	0.11828238993457\\
-0.987156160081357	0.130462520259816\\
-0.984197650866276	0.145255066335218\\
-0.981976569024304	0.155249934624092\\
-0.981115849644926	0.159037099893356\\
-0.980194038548591	0.163000887607598\\
-0.979205404824194	0.167153149250064\\
-0.978143553328547	0.171506740382219\\
-0.977001333895545	0.176075618114224\\
-0.975770736456098	0.18087494812807\\
-0.974442769681899	0.185921221870024\\
-0.973007320340543	0.191232384433041\\
-0.971452990063687	0.196827973429723\\
-0.969766905658068	0.202729268849389\\
-0.967934498439619	0.208959453392117\\
-0.965939247337516	0.215543782029056\\
-0.963762379697734	0.22250975847636\\
-0.961382522847226	0.229887314712933\\
-0.958775298587157	0.237708987493141\\
-0.955912851950104	0.246010082740595\\
-0.952763304909315	0.254828814454803\\
-0.949290125479817	0.264206398914448\\
-0.945451403138446	0.274187077002014\\
-0.941199023235727	0.28481802675881\\
-0.936477736865097	0.296149114048322\\
-0.931224129655156	0.308232410631187\\
-0.925365504824573	0.32112138525847\\
-0.918818714900306	0.33486964409943\\
-0.911489005798657	0.349529062302728\\
-0.903268980286513	0.365147110775802\\
-0.894037849306398	0.381763146540008\\
-0.883661222737217	0.399403411707617\\
-0.871991796429096	0.41807449380061\\
-0.858871413280568	0.437755068523402\\
-0.844135092308348	0.45838591788451\\
-0.827617687479945	0.479858544161434\\
-0.809163781139607	0.502003231769839\\
-0.788641123171331	0.524578155534943\\
-0.765957266520489	0.547262012185785\\
-0.741077929952977	0.569653415096546\\
-0.714044088461155	0.591280488290004\\
-0.684983161622723	0.611623137076906\\
-0.654108557335916	0.63014789964899\\
-0.621701971532707	0.646351192550595\\
-0.588074617118414	0.659802134316312\\
-0.553506366597155	0.67017260947269\\
-0.518163921687246	0.677241098454671\\
-0.481998844841091	0.680857606139287\\
-0.444626957256001	0.680857606139287\\
-0.405227964998137	0.676917706913501\\
-0.36282270849924	0.668436655613721\\
-0.318745292882476	0.655213430928692\\
0	0.5\\
};
\addplot [color=red]
  table[row sep=crcr]{%
-1	-0\\
-0.996336310949872	-0.062238958115719\\
-0.995828471850466	-0.0665555904606681\\
-0.995327905223371	-0.0705601234774225\\
-0.99508728455282	-0.0723647785065564\\
-0.993189253292471	-0.0856509973290036\\
-0.992790731223931	-0.0882413907745141\\
-0.991099961581726	-0.0983860086277391\\
-0.98941365292432	-0.107660706243475\\
-0.986887096794019	-0.120293486894978\\
-0.985189325789641	-0.127933456414682\\
-0.985046838061473	-0.128560402418618\\
-0.982082002349273	-0.141309195981081\\
-0.98336265377989	-0.135930459972486\\
-0.98060132206473	-0.147251920004643\\
-0.987138027079645	-0.121105099944984\\
-0.972884964467079	-0.176692044133991\\
-0.978431187956769	-0.15561639487317\\
-0.976029564867108	-0.164502400304914\\
-0.97815488368619	-0.15685125255622\\
-0.972664446554909	-0.176067782515705\\
-0.972687211104474	-0.175990383047183\\
-0.969926767344659	-0.185099847454572\\
-0.969271921472229	-0.187195354246347\\
-0.966043201844905	-0.197204385091052\\
-0.963858494951827	-0.203758505770286\\
-0.961701621584624	-0.210013438535174\\
-0.959063778894355	-0.217399398067929\\
-0.954579670961268	-0.229506489487264\\
-0.951522540307591	-0.237455029186822\\
-0.94477726858765	-0.254318208486676\\
-0.945214336096031	-0.253269246466561\\
-0.9380446768193	-0.269759462803044\\
-0.931686694853226	-0.283747023128406\\
-0.925385417238695	-0.296979706118921\\
-0.916961708436207	-0.313827123723896\\
-0.911145904660215	-0.324877150898282\\
-0.900522683408927	-0.343998949150599\\
-0.887799507396145	-0.36562834837233\\
-0.876233166509629	-0.384134493790754\\
-0.862293431499978	-0.405044096305232\\
-0.8488131944598	-0.42391642816148\\
-0.82782687935318	-0.451198637800086\\
-0.8089306209234	-0.473874147915823\\
-0.783424145677049	-0.501931270686809\\
-0.761205071979478	-0.52415034438438\\
-0.733790595303114	-0.548823373393108\\
-0.703837142336352	-0.572786135766517\\
-0.672485348305038	-0.594732391588437\\
-0.639650546742377	-0.614433272526033\\
-0.604645495380269	-0.631935798207087\\
-0.572094006386093	-0.644956393804758\\
-0.535059962803386	-0.65606660687957\\
-0.500177236406749	-0.663043152158897\\
-0.463942123141224	-0.66666666348545\\
-0.425304850932319	-0.66666666348545\\
-0.383491660365067	-0.662485344428725\\
-0.338770766925158	-0.653541165740743\\
-0.288185715360556	-0.638365650271362\\
-0.231754017153009	-0.615792970988343\\
0	-0.5\\
};\addlegendentry{$\LocTI{5}$}
\addplot [color=red, forget plot]
  table[row sep=crcr]{%
-1	0\\
-0.996336310949872	0.062238958115719\\
-0.995828471850466	0.0665555904606681\\
-0.995327905223371	0.0705601234774225\\
-0.99508728455282	0.0723647785065564\\
-0.993189253292471	0.0856509973290036\\
-0.992790731223931	0.0882413907745141\\
-0.991099961581726	0.0983860086277391\\
-0.98941365292432	0.107660706243475\\
-0.986887096794019	0.120293486894978\\
-0.985189325789641	0.127933456414682\\
-0.985046838061473	0.128560402418618\\
-0.982082002349273	0.141309195981081\\
-0.98336265377989	0.135930459972486\\
-0.98060132206473	0.147251920004643\\
-0.987138027079645	0.121105099944984\\
-0.972884964467079	0.176692044133991\\
-0.978431187956769	0.15561639487317\\
-0.976029564867108	0.164502400304914\\
-0.97815488368619	0.15685125255622\\
-0.972664446554909	0.176067782515705\\
-0.972687211104474	0.175990383047183\\
-0.969926767344659	0.185099847454572\\
-0.969271921472229	0.187195354246347\\
-0.966043201844905	0.197204385091052\\
-0.963858494951827	0.203758505770286\\
-0.961701621584624	0.210013438535174\\
-0.959063778894355	0.217399398067929\\
-0.954579670961268	0.229506489487264\\
-0.951522540307591	0.237455029186822\\
-0.94477726858765	0.254318208486676\\
-0.945214336096031	0.253269246466561\\
-0.9380446768193	0.269759462803044\\
-0.931686694853226	0.283747023128406\\
-0.925385417238695	0.296979706118921\\
-0.916961708436207	0.313827123723896\\
-0.911145904660215	0.324877150898282\\
-0.900522683408927	0.343998949150599\\
-0.887799507396145	0.36562834837233\\
-0.876233166509629	0.384134493790754\\
-0.862293431499978	0.405044096305232\\
-0.8488131944598	0.42391642816148\\
-0.82782687935318	0.451198637800086\\
-0.8089306209234	0.473874147915823\\
-0.783424145677049	0.501931270686809\\
-0.761205071979478	0.52415034438438\\
-0.733790595303114	0.548823373393108\\
-0.703837142336352	0.572786135766517\\
-0.672485348305038	0.594732391588437\\
-0.639650546742377	0.614433272526033\\
-0.604645495380269	0.631935798207087\\
-0.572094006386093	0.644956393804758\\
-0.535059962803386	0.65606660687957\\
-0.500177236406749	0.663043152158897\\
-0.463942123141224	0.66666666348545\\
-0.425304850932319	0.66666666348545\\
-0.383491660365067	0.662485344428725\\
-0.338770766925158	0.653541165740743\\
-0.288185715360556	0.638365650271362\\
-0.231754017153009	0.615792970988343\\
0	0.5\\
};
\addplot [color=blue]
  table[row sep=crcr]{%
-1	-0\\
-0.996376337347805	-0.0617359930861756\\
-0.994922823125307	-0.0719105926436558\\
-0.993709359312353	-0.079798107427858\\
-0.992373395801971	-0.0878138884901487\\
-0.990720892491753	-0.0969026566963498\\
-0.988570472583691	-0.107654756236662\\
-0.98697438213485	-0.114837163256444\\
-0.986357773036909	-0.117550243287384\\
-0.985697475369149	-0.120389523258755\\
-0.984989016128747	-0.123365052068441\\
-0.984227525349514	-0.126487164263298\\
-0.983407321922466	-0.129767977971488\\
-0.982522234324485	-0.133219819603616\\
-0.981564844520325	-0.136857900859422\\
-0.980526268974912	-0.140700630377451\\
-0.979395585803821	-0.144771089793379\\
-0.978175626365289	-0.149040947828241\\
-0.976882245270584	-0.153438443550237\\
-0.975478106011892	-0.15807210310392\\
-0.973951607919398	-0.162956896999902\\
-0.972299356949579	-0.168078875006341\\
-0.970531401020435	-0.173382742793773\\
-0.968685503687991	-0.178735845057861\\
-0.966629644404956	-0.184492251050358\\
-0.963720491440951	-0.192346964053173\\
-0.960383568780538	-0.201022962970247\\
-0.956874814485226	-0.209794848708526\\
-0.952976202540627	-0.219151517375564\\
-0.948552570842159	-0.229325870282039\\
-0.943490355257927	-0.24046274456735\\
-0.937665017730941	-0.25269595337402\\
-0.93093019119661	-0.266165606442683\\
-0.923112489365118	-0.281019239922517\\
-0.914005504774288	-0.297411812186012\\
-0.903363560078224	-0.315503118169319\\
-0.890897246936887	-0.335449219195459\\
-0.876277348111176	-0.357379067434026\\
-0.859158125358749	-0.381345979287423\\
-0.839227698410295	-0.407255534320414\\
-0.816249922796263	-0.434828865057252\\
-0.790057788275242	-0.463640213030376\\
-0.760624403725957	-0.49307359757966\\
-0.728390978845159	-0.522083679972379\\
-0.694189079455663	-0.549445199483976\\
-0.658832804830298	-0.574194591721731\\
-0.622961873071609	-0.595717150776945\\
-0.586907569962237	-0.61374430233163\\
-0.55093943284661	-0.628131557177881\\
-0.515214374691471	-0.638849074624423\\
-0.479075824694063	-0.646076784623904\\
-0.441651633159366	-0.649819203777374\\
-0.402009603370538	-0.649819203777374\\
-0.358855203364394	-0.64550376377676\\
-0.310250579690302	-0.635782839041941\\
-0.251994632507815	-0.618306054887195\\
0	-0.5\\
};\addlegendentry{$\WMTI{5}$}
\addplot [color=blue, forget plot]
  table[row sep=crcr]{%
-1	0\\
-0.996376337347805	0.0617359930861756\\
-0.994922823125307	0.0719105926436558\\
-0.993709359312353	0.079798107427858\\
-0.992373395801971	0.0878138884901487\\
-0.990720892491753	0.0969026566963498\\
-0.988570472583691	0.107654756236662\\
-0.98697438213485	0.114837163256444\\
-0.986357773036909	0.117550243287384\\
-0.985697475369149	0.120389523258755\\
-0.984989016128747	0.123365052068441\\
-0.984227525349514	0.126487164263298\\
-0.983407321922466	0.129767977971488\\
-0.982522234324485	0.133219819603616\\
-0.981564844520325	0.136857900859422\\
-0.980526268974912	0.140700630377451\\
-0.979395585803821	0.144771089793379\\
-0.978175626365289	0.149040947828241\\
-0.976882245270584	0.153438443550237\\
-0.975478106011892	0.15807210310392\\
-0.973951607919398	0.162956896999902\\
-0.972299356949579	0.168078875006341\\
-0.970531401020435	0.173382742793773\\
-0.968685503687991	0.178735845057861\\
-0.966629644404956	0.184492251050358\\
-0.963720491440951	0.192346964053173\\
-0.960383568780538	0.201022962970247\\
-0.956874814485226	0.209794848708526\\
-0.952976202540627	0.219151517375564\\
-0.948552570842159	0.229325870282039\\
-0.943490355257927	0.24046274456735\\
-0.937665017730941	0.25269595337402\\
-0.93093019119661	0.266165606442683\\
-0.923112489365118	0.281019239922517\\
-0.914005504774288	0.297411812186012\\
-0.903363560078224	0.315503118169319\\
-0.890897246936887	0.335449219195459\\
-0.876277348111176	0.357379067434026\\
-0.859158125358749	0.381345979287423\\
-0.839227698410295	0.407255534320414\\
-0.816249922796263	0.434828865057252\\
-0.790057788275242	0.463640213030376\\
-0.760624403725957	0.49307359757966\\
-0.728390978845159	0.522083679972379\\
-0.694189079455663	0.549445199483976\\
-0.658832804830298	0.574194591721731\\
-0.622961873071609	0.595717150776945\\
-0.586907569962237	0.61374430233163\\
-0.55093943284661	0.628131557177881\\
-0.515214374691471	0.638849074624423\\
-0.479075824694063	0.646076784623904\\
-0.441651633159366	0.649819203777374\\
-0.402009603370538	0.649819203777374\\
-0.358855203364394	0.64550376377676\\
-0.310250579690302	0.635782839041941\\
-0.251994632507815	0.618306054887195\\
0	0.5\\
};

\addplot [color=green]
  table[row sep=crcr]{%
-1	-0\\
-0.996740044596248	-0.0570959429538203\\
-0.996333679141777	-0.0605500493168165\\
-0.995846372217816	-0.0644485047085084\\
-0.995255128978792	-0.0688828290011854\\
-0.994528178066545	-0.073971485386914\\
-0.993620627793929	-0.0798705621589217\\
-0.992467394831001	-0.0867899599364872\\
-0.990971261443658	-0.0950186935668755\\
-0.988981822700316	-0.104965887283586\\
-0.987493321578427	-0.111664142332087\\
-0.986917437093199	-0.114198034067088\\
-0.986300992289188	-0.116848746724338\\
-0.985640102660378	-0.119624483165337\\
-0.98493040943381	-0.122534225394268\\
-0.984167008711173	-0.125587828284814\\
-0.983344368016853	-0.128796126992662\\
-0.982456227647867	-0.132171060394811\\
-0.981495483600127	-0.135725813371448\\
-0.980454048064709	-0.139474981298954\\
-0.979322682487091	-0.143434760820617\\
-0.978090796904517	-0.147623171801366\\
-0.976746207627527	-0.152060316415434\\
-0.975274843196007	-0.156768682596298\\
-0.973660385758403	-0.161773500652871\\
-0.971883831371611	-0.167103163813247\\
-0.969922947911756	-0.172789725846827\\
-0.967751602901345	-0.178869491875978\\
-0.96533892504615	-0.185383722085004\\
-0.9626482518521	-0.192379472389534\\
-0.95963580029005	-0.199910601294659\\
-0.95624897662351	-0.208038978094355\\
-0.952424213197361	-0.216835933974496\\
-0.948084181520596	-0.226384003663381\\
-0.943134178896278	-0.236779009174448\\
-0.937457416035116	-0.248132534896772\\
-0.930908841523716	-0.260574826468432\\
-0.92330702436458	-0.274258097354876\\
-0.914423486518045	-0.289360111693987\\
-0.903968770627796	-0.306087657118385\\
-0.891574558503178	-0.324678975305312\\
-0.876771641023046	-0.345403059777497\\
-0.858965325246387	-0.368551270287154\\
-0.837415168064785	-0.394411458905076\\
-0.811240518246083	-0.423203573705648\\
-0.779511808155452	-0.454932283796279\\
-0.741582328273131	-0.489068815690368\\
-0.697997775116377	-0.523936458215771\\
-0.652228024483386	-0.555975283658865\\
-0.610177372290058	-0.581205674974862\\
-0.572366965474402	-0.60011087838269\\
-0.536058445350093	-0.614634286432413\\
-0.499973657312627	-0.625459722843653\\
-0.463258462054345	-0.632802761895309\\
-0.425088357331626	-0.636619772367581\\
-0.384464129821992	-0.636619772367581\\
-0.339935535414198	-0.632166912926802\\
-0.288985491357354	-0.621976904115433\\
-0.226019031174655	-0.603086966060623\\
-0.126793535907613	-0.563396767953806\\
0	-0.5\\
};\addlegendentry{$\TI$}
\addplot [color=green, forget plot]
  table[row sep=crcr]{%
-1	0\\
-0.996740044596248	0.0570959429538203\\
-0.996333679141777	0.0605500493168165\\
-0.995846372217816	0.0644485047085084\\
-0.995255128978792	0.0688828290011854\\
-0.994528178066545	0.073971485386914\\
-0.993620627793929	0.0798705621589217\\
-0.992467394831001	0.0867899599364872\\
-0.990971261443658	0.0950186935668755\\
-0.988981822700316	0.104965887283586\\
-0.987493321578427	0.111664142332087\\
-0.986917437093199	0.114198034067088\\
-0.986300992289188	0.116848746724338\\
-0.985640102660378	0.119624483165337\\
-0.98493040943381	0.122534225394268\\
-0.984167008711173	0.125587828284814\\
-0.983344368016853	0.128796126992662\\
-0.982456227647867	0.132171060394811\\
-0.981495483600127	0.135725813371448\\
-0.980454048064709	0.139474981298954\\
-0.979322682487091	0.143434760820617\\
-0.978090796904517	0.147623171801366\\
-0.976746207627527	0.152060316415434\\
-0.975274843196007	0.156768682596298\\
-0.973660385758403	0.161773500652871\\
-0.971883831371611	0.167103163813247\\
-0.969922947911756	0.172789725846827\\
-0.967751602901345	0.178869491875978\\
-0.96533892504615	0.185383722085004\\
-0.9626482518521	0.192379472389534\\
-0.95963580029005	0.199910601294659\\
-0.95624897662351	0.208038978094355\\
-0.952424213197361	0.216835933974496\\
-0.948084181520596	0.226384003663381\\
-0.943134178896278	0.236779009174448\\
-0.937457416035116	0.248132534896772\\
-0.930908841523716	0.260574826468432\\
-0.92330702436458	0.274258097354876\\
-0.914423486518045	0.289360111693987\\
-0.903968770627796	0.306087657118385\\
-0.891574558503178	0.324678975305312\\
-0.876771641023046	0.345403059777497\\
-0.858965325246387	0.368551270287154\\
-0.837415168064785	0.394411458905076\\
-0.811240518246083	0.423203573705648\\
-0.779511808155452	0.454932283796279\\
-0.741582328273131	0.489068815690368\\
-0.697997775116377	0.523936458215771\\
-0.652228024483386	0.555975283658865\\
-0.610177372290058	0.581205674974862\\
-0.572366965474402	0.60011087838269\\
-0.536058445350093	0.614634286432413\\
-0.499973657312627	0.625459722843653\\
-0.463258462054345	0.632802761895309\\
-0.425088357331626	0.636619772367581\\
-0.384464129821992	0.636619772367581\\
-0.339935535414198	0.632166912926802\\
-0.288985491357354	0.621976904115433\\
-0.226019031174655	0.603086966060623\\
-0.126793535907613	0.563396767953806\\
0	0.5\\
};
\addplot [color=black, forget plot]
  table[row sep=crcr]{%
-1	1\\
1	0\\
};
\addplot [color=black, forget plot]
  table[row sep=crcr]{%
-1	1\\
-1	-1\\
};
\addplot [color=black, forget plot]
  table[row sep=crcr]{%
-1	-1\\
1	0\\
};

\end{axis}

\end{tikzpicture}%

%% file: energyErrorsAni.tex
%
%
\begin{tikzpicture}

\begin{axis}[%
width=5.4in,
height=3.5in,
scale only axis,
xmin=-9,
xmax=1,
xlabel=$\Delta$,
xtick={-8,-6,-4,-2,0},
ymode=log,
ymin=0.001,
ymax=0.277226277372263,
yminorticks=true,
ylabel=$e_0-\hat e$,
axis background/.style={fill=white},
axis x line*=bottom,
axis y line*=left
]
\addplot [color=blue]
  table[row sep=crcr]{%
-9	0.0375208115969095\\
-8.8	0.0386241621077339\\
-8.6	0.0397844795973885\\
-8.4	0.0410058703455078\\
-8.2	0.0422928099127446\\
-8	0.0436501811337973\\
-7.8	0.0450833160969228\\
-7.6	0.0465980424090739\\
-7.4	0.0482007340022674\\
-7.2	0.0498983666541841\\
-7	0.0516985782543884\\
-6.8	0.0536097336201928\\
-6.6	0.0556409933131725\\
-6.4	0.0578023853711747\\
-6.2	0.0601048780708293\\
-6	0.0625604506530646\\
-5.8	0.0651821572099553\\
-5.6	0.0679841764022857\\
-5.4	0.0709818360006915\\
-5.2	0.0741915959139581\\
-5	0.0776309656566889\\
-4.8	0.0813183210675534\\
-4.6	0.085272569019236\\
-4.4	0.0895125857276771\\
-4.2	0.0940563210668817\\
-4	0.0989194138660761\\
-3.8	0.104113095913368\\
-3.6	0.109641068160808\\
-3.4	0.115494903172445\\
-3.2	0.12164735592861\\
-3	0.12804275076835\\
-2.8	0.134583379812872\\
-2.6	0.141110689555812\\
-2.4	0.1473802171091\\
-2.2	0.153030494494136\\
-2	0.157550373509086\\
-1.8	0.160261281836079\\
-1.6	0.160362161500109\\
-1.4	0.157152898831154\\
-1.2	0.150601871403287\\
-1	0.141815634394115\\
-0.800000000000001	0.131948633182454\\
-0.6	0.12154539882879\\
-0.4	0.110838856971885\\
-0.199999999999999	0.0998577650453043\\
0	0.0884756675434113\\
0.199999999999999	0.0764431020566096\\
0.4	0.0633846161398217\\
0.6	0.0486171727114446\\
0.800000000000001	0.0300719203698803\\
}node[above,pos=0.3]{$l=3$};\addlegendentry{$\WM{l}$}
\addplot [color=blue, forget plot]
  table[row sep=crcr]{%
-9	0.00549020157934166\\
-8.8	0.00576233236239787\\
-8.6	0.00605221479555418\\
-8.4	0.00636106137655545\\
-8.2	0.00669023564822879\\
-8	0.00704134769250864\\
-7.8	0.00741626549717367\\
-7.6	0.00781725155358082\\
-7.4	0.00824707455317597\\
-7.2	0.00870915127174943\\
-7	0.0092076850980467\\
-6.8	0.00974780304335088\\
-6.6	0.0103356556933072\\
-6.4	0.0109784784682505\\
-6.2	0.0116846161962059\\
-6	0.0124635382926197\\
-5.8	0.0133258831977141\\
-5.6	0.0142835201873384\\
-5.4	0.0153497044631985\\
-5.2	0.0165392508383837\\
-5	0.0178687399458548\\
-4.8	0.0193567298226629\\
-4.6	0.0210239374367012\\
-4.4	0.0228934153895928\\
-4.2	0.0249909037526743\\
-4	0.0273457065646081\\
-3.8	0.0299929156746579\\
-3.6	0.032977535951602\\
-3.4	0.0363574387564141\\
-3.2	0.040194750987089\\
-3	0.0445274854224382\\
-2.8	0.049339124042366\\
-2.6	0.0545495557883657\\
-2.4	0.0600100537724537\\
-2.2	0.0654716908668118\\
-2	0.0705342951551113\\
-1.8	0.0746081648309964\\
-1.6	0.0769399121607819\\
-1.4	0.076819874725595\\
-1.2	0.0741420940233122\\
-1	0.069880519221436\\
-0.800000000000001	0.0650216127884413\\
-0.6	0.0599426280921789\\
-0.4	0.05475080141836\\
-0.199999999999999	0.0494254170554373\\
0	0.0438847550199803\\
0.199999999999999	0.0380150018519922\\
0.4	0.0316977475823705\\
0.6	0.024571018730926\\
0.800000000000001	0.0157602419414316\\
}node[above,pos=0.3]{$l=4$};
\addplot [color=blue, forget plot]
  table[row sep=crcr]{%
-9	0.0016755868565248\\
-8.8	0.00177937474524015\\
-8.6	0.00189130755649813\\
-8.4	0.00201201094050596\\
-8.2	0.00214222824683219\\
-8	0.00228280506369138\\
-7.8	0.00243474242143282\\
-7.6	0.00259916915990654\\
-7.4	0.00277744578541483\\
-7.2	0.00297128871045782\\
-7	0.00318298116263538\\
-6.8	0.00341239238699576\\
-6.6	0.00365410271378241\\
-6.4	0.00390772303690934\\
-6.2	0.00417237009223115\\
-6	0.00444457585399594\\
-5.8	0.0047150619242311\\
-5.6	0.00496255076898411\\
-5.4	0.00518694246826179\\
-5.2	0.00551062918930167\\
-5	0.00596356580361412\\
-4.8	0.0065157629645789\\
-4.6	0.00717031778115551\\
-4.4	0.00794464625219593\\
-4.2	0.00886341150472969\\
-4	0.00995724373779705\\
-3.8	0.0112626887054983\\
-3.6	0.0128219591372178\\
-3.4	0.0146822630552763\\
-3.2	0.0168942483432404\\
-3	0.0195085530814221\\
-2.8	0.0225679951598226\\
-2.6	0.0260906982926818\\
-2.4	0.0300382236599002\\
-2.2	0.0342712727136045\\
-2	0.0385078187077728\\
-1.8	0.0422852995936718\\
-1.6	0.0449235694792662\\
-1.4	0.0456853086114091\\
-1.2	0.0443643525420268\\
-1	0.0418074523857166\\
-0.800000000000001	0.0388993314881496\\
-0.6	0.0359231339888462\\
-0.4	0.0328749905130774\\
-0.199999999999999	0.0297115179851339\\
0	0.0263988628195859\\
0.199999999999999	0.0228897681098768\\
0.4	0.0191058345198376\\
0.6	0.014852816853248\\
0.800000000000001	0.00965769658661597\\
}node[above,pos=0.35]{$l=5$};
\addplot [color=red]
  table[row sep=crcr]{%
-9	0.106222387158706\\
-8.8	0.108418488762368\\
-8.6	0.1107020919672\\
-8.4	0.113078132554625\\
-8.2	0.115551880294781\\
-8	0.118128960860432\\
-7.8	0.120815378021624\\
-7.6	0.123617535629303\\
-7.4	0.126542258683864\\
-7.2	0.129596812498948\\
-7	0.132788918584515\\
-6.8	0.136126765350946\\
-6.6	0.139619011030705\\
-6.4	0.143274775260138\\
-6.2	0.147103614472326\\
-6	0.151115490635583\\
-5.8	0.155320623782263\\
-5.6	0.159729478375649\\
-5.4	0.164352500097788\\
-5.2	0.169199870791313\\
-5	0.174281126519999\\
-4.8	0.179604618498104\\
-4.6	0.185176758362878\\
-4.4	0.191000967282927\\
-4.2	0.197076217971295\\
-4	0.203395016608424\\
-3.8	0.209940613742723\\
-3.6	0.216683154115728\\
-3.4	0.223574369125983\\
-3.2	0.230540277644553\\
-3	0.237471193460868\\
-2.8	0.244208165644157\\
-2.6	0.250524889828155\\
-2.4	0.256104381699533\\
-2.2	0.260511001965388\\
-2	0.263162625645832\\
-1.8	0.263319640481026\\
-1.6	0.260138325152801\\
-1.4	0.252903582711268\\
-1.2	0.241604953146471\\
-1	0.227411277889189\\
-0.800000000000001	0.211578516576653\\
-0.6	0.194779589426671\\
-0.4	0.177376777105443\\
-0.199999999999999	0.159487471675893\\
0	0.140974017621277\\
0.199999999999999	0.121397968976868\\
0.4	0.0999260729327862\\
0.6	0.0751207160636167\\
0.800000000000001	0.0443351384825079\\
}node[above,pos=0.3]{$l=3,4$};\addlegendentry{$\LocTI{l}$}
\addplot [color=red, forget plot]
  table[row sep=crcr]{%
-9	0.0118348230327161\\
-8.8	0.0123534839632278\\
-8.6	0.0129062444102797\\
-8.4	0.0134960787201841\\
-8.2	0.0141262800721176\\
-8	0.014800501345217\\
-7.8	0.0155228009916817\\
-7.6	0.0162976952219047\\
-7.4	0.0171302203957326\\
-7.2	0.0180259816091253\\
-7	0.0189912654137219\\
-6.8	0.0200331002754446\\
-6.6	0.0211593549076934\\
-6.4	0.0223788582855624\\
-6.2	0.0237015204092001\\
-6	0.0251384692031786\\
-5.8	0.0267022007377689\\
-5.6	0.0284067390764928\\
-5.4	0.0302678026241967\\
-5.2	0.0323029664263856\\
-5	0.0345318084605717\\
-4.8	0.0369760114741702\\
-4.6	0.0396593860198893\\
-4.4	0.042607747993892\\
-4.2	0.0458485588101407\\
-4	0.0494101798232345\\
-3.8	0.0533205172530451\\
-3.6	0.0576047281738097\\
-3.4	0.0622814945584547\\
-3.2	0.0673571625995568\\
-3	0.0728167622652864\\
-2.8	0.0786105983786509\\
-2.6	0.0846348587036365\\
-2.4	0.0907047665596226\\
-2.2	0.0965200452842527\\
-2	0.10162682414873\\
-1.8	0.105392527635534\\
-1.6	0.107042313992142\\
-1.4	0.105874466327086\\
-1.2	0.101823008956193\\
-1	0.0959283697672839\\
-0.800000000000001	0.0892561162941412\\
-0.6	0.082256228832108\\
-0.4	0.0750841848533592\\
-0.199999999999999	0.0677318397329301\\
0	0.060093788719207\\
0.199999999999999	0.0520029722263244\\
0.4	0.0432142496161354\\
0.6	0.0332536958369478\\
0.800000000000001	0.0208239508747012\\
}node[above,pos=0.3]{$l=5,6$};
\addplot [color=red, forget plot]
  table[row sep=crcr]{%
-9	0.0016320384975792\\
-8.8	0.00174180308016503\\
-8.6	0.00186151803963064\\
-8.4	0.00199232131353888\\
-8.2	0.00213550763417203\\
-8	0.00229257064770216\\
-7.8	0.00246514914776608\\
-7.6	0.00265523055030314\\
-7.4	0.00286502277305889\\
-7.2	0.0030970997345614\\
-7	0.00335442915839845\\
-6.8	0.00364044869699676\\
-6.6	0.00395916622070036\\
-6.4	0.0043152462801217\\
-6.2	0.00471414135035175\\
-6	0.00516225056097852\\
-5.8	0.0056670909850407\\
-5.6	0.00623751958239716\\
-5.4	0.0068839929172011\\
-5.2	0.00761887908023962\\
-5	0.00845682003324555\\
-4.8	0.00941515515131108\\
-4.6	0.010514397216336\\
-4.4	0.0117787477159208\\
-4.2	0.0132366103678114\\
-4	0.0149210234166741\\
-3.8	0.0168698565756582\\
-3.6	0.0191254970519967\\
-3.4	0.0217335472089353\\
-3.2	0.0247397443188522\\
-3	0.0281838408469546\\
-2.8	0.0320885545184542\\
-2.6	0.0364410002672493\\
-2.4	0.0411636094059675\\
-2.2	0.046072517796635\\
-2	0.0508261683808451\\
-1.8	0.054881053685325\\
-1.6	0.0575066286528862\\
-1.4	0.0579839854113908\\
-1.2	0.0561649772721464\\
-1	0.0529579822356834\\
-0.800000000000001	0.0492780443998784\\
-0.6	0.0454439436141159\\
-0.4	0.041530098364245\\
-0.199999999999999	0.0375112800395576\\
0	0.0333234201413468\\
0.199999999999999	0.028883041060574\\
0.4	0.0240707484024167\\
0.6	0.0186597036534684\\
0.800000000000001	0.0120137367996493\\
}node[above,pos=0.26]{$l=7,8$};
\end{axis}

\end{tikzpicture}%